\documentclass[table,12pt]{article}
\usepackage[utf8]{inputenc}
\usepackage{amsmath}

\usepackage{array}
\usepackage{multirow}
\usepackage{booktabs} 
\usepackage{bm}
\usepackage{mathtools}
\usepackage{amssymb}
\usepackage{accents}
\usepackage{amsthm}
\usepackage{natbib}
\usepackage[english]{babel} 
\usepackage[protrusion=true,expansion=true]{microtype} 
\usepackage{microtype} 
\usepackage{tabularx}
\usepackage[table]{xcolor}
\usepackage{tikz}
\usetikzlibrary{positioning}
\usepackage[font = small,labelfont=bf,textfont=it]{caption} 
\usepackage[toc,page]{appendix}
\usepackage{dsfont}
\DeclareSymbolFont{bbold}{U}{bbold}{m}{n}
\DeclareSymbolFontAlphabet{\mathbbold}{bbold} 

\usepackage{natbib}
\usepackage{setspace}
\usepackage{soul}
\usepackage{enumitem}

\usepackage{multirow}
\usepackage{algorithm}
\usepackage[noend]{algpseudocode}
\usepackage{etoolbox}
\algblock{ParFor}{EndParFor}
\algnewcommand\algorithmicparfor{\textbf{for}}
\algnewcommand\algorithmicpardo{}
\algnewcommand\algorithmicendparfor{}
\algrenewtext{ParFor}[1]{\algorithmicparfor\ #1\ }
\algrenewtext{EndParFor}{\algorithmicendparfor}
\newcolumntype{C}[1]{>{\centering\let\newline\\\arraybackslash\hspace{0pt}}m{#1}}

\newcommand{\argmax}{\operatorname*{argmax}}

\usepackage{algorithm}
\usepackage{algpseudocode}

\usepackage{graphicx,epstopdf}
\makeatletter
\newcommand{\distas}[1]{\mathbin{\overset{#1}{\kern\z@\sim}}}%
\newsavebox{\mybox}\newsavebox{\mysim}
\newtheorem{theorem}{Theorem}

\newcommand{\distras}[1]{%
  \savebox{\mybox}{\hbox{\kern3pt$\scriptstyle#1$\kern3pt}}%
  \savebox{\mysim}{\hbox{$\sim$}}%
  \mathbin{\overset{#1}{\kern\z@\resizebox{\wd\mybox}{\ht\mysim}{$\sim$}}}%
}

\newcommand{\blind}{1}

\begin{document}

\def\spacingset#1{\renewcommand{\baselinestretch}%
{#1}\small\normalsize} \spacingset{1.5}


\if1\blind
{
 \centering{\bf\Large Active Learning with Adaptive Non-Stationary Kernel for Continuous-Fidelity Surrogate Models
 }\\
 
  \vspace{0.3in}
  \centering{Romain Boutelet and Chih-Li Sung\footnote{These authors gratefully acknowledge funding from NSF DMS 2338018.}\vspace{0.1in}\\
        Michigan State University\\
        }
    \date{\vspace{-7ex}}
} \fi

\if0\blind
{
  \bigskip
  \bigskip
  \bigskip
    \begin{center}
    {\Large\bf Active Learning with Adaptive Non-Stationary Kernel for Continuous-Fidelity Surrogate Models}
    \end{center}
  \medskip
} \fi

\bigskip
\begin{abstract}
    Simulating complex physical processes across a domain of input parameters can be very computationally expensive. Multi-fidelity surrogate modeling can resolve this issue by integrating cheaper simulations with the expensive ones in order to obtain better predictions at a reasonable cost. We are specifically interested in computer experiments where real-valued fidelity parameters determine the fidelity of the numerical output, such as finite element methods. In these cases, integrating this fidelity parameter in the analysis enables us to 
    make inference on fidelity levels that have not been observed yet. Such models have been developed, and we propose a new adaptive non-stationary kernel function which more accurately reflects the behavior of computer simulation outputs. In addition, we develop an active learning strategy based on the integrated mean squared prediction error (IMSPE) to identify the best design points across input parameters and fidelity parameters, while taking into account the computational cost associated with the fidelity parameters. We illustrate this methodology through numerical examples and applications to finite element methods. An \textsf{R} package for the proposed methodology is provided in an open repository.
\end{abstract}

\noindent%
{\it Keywords}: Sequential design; Multi-fidelity model; Uncertainty quantification; Finite element methods; Gaussian process.
\vfill

\newpage
\spacingset{1} 

\section{Introduction}

Computer simulations are pivotal across diverse scientific fields ranging from plant biology to aerospace engineering. High-fidelity simulations, while highly accurate, often require substantial computational resources and time. For instance, a large-eddy simulation of a single injector design in \cite{mak2018efficient} can take six days of computation, even when parallelized over 200 CPU cores.

To address these computational challenges, multi-fidelity simulations have emerged as a practical solution. These simulations leverage a combination of low-fidelity and high-fidelity models. Low-fidelity simulations are computationally cheaper but less accurate, whereas high-fidelity simulations offer greater accuracy at a higher computational cost. A statistical model, or \textit{emulator}, integrates data from both fidelity levels to generate accurate predictions efficiently. When simulations are costly to run, emulators become indispensable tools, particularly for exploring a wide range of input variables. The foundational work by \cite{kennedy2000predicting} has significantly impacted this approach by modeling a sequence of computer simulations from lowest to highest fidelity using a series of stationary Gaussian process (GP) models linked by a linear auto-regressive framework. Subsequent developments inspired by this approach include \cite{qian2006building}, \cite{le2013bayesian}, \cite{le2014recursive}, \cite{qian2008bayesian}, \cite{perdikaris2017nonlinear}, \cite{ji2022graphica}, and \cite{heo2023active}, among others.

Finite element methods (FEM) exemplify the need for multi-fidelity approaches \citep{brenner2007fem}. FEM are widely used to simulate real-world phenomena such as soil erosion and climate change, involving fidelity parameters like mesh density that control numerical accuracy and computational cost. Coarser meshes in FEM are less costly but less accurate, whereas finer meshes provide higher accuracy at increased computational expense. Instead of solely relying on the ordering of the levels, \cite{tuo2014surrogate} first incorporated the fidelity parameters as scalars in the statistical model, introducing a \textit{non-stationary} GP model consistent with numerical analysis results, and enabling the prediction of the exact solution. This model has been further developed for conglomerate multi-fidelity emulation \citep{ji2022multi}.

Despite these advancements, important gaps remain in the \textit{design} of multi-fidelity experiments with continuous fidelity parameters for cost-efficient emulation. Notable exceptions include the recent work by \cite{shaowu2023design}, which introduced one-shot, batch designs for multi-fidelity finite element simulations based on the model of \cite{tuo2014surrogate}, and the work by \cite{sung2022stacking}, which proposed sequential batch designs based on the auto-regressive model of \cite{kennedy2000predicting}.

This study addresses these critical gaps through two critical contributions. First, we propose a novel, \textit{adaptive} non-stationary kernel function, inspired by fractional Brownian motion, to enhance the surrogate model. The increased flexibility in this new kernel accommodates varying degrees of correlation in simulation increments along the fidelity parameters, making it more well suited particularly  for multi-fidelity finite element simulations. This enhancement not only improves prediction performance, but also ensures better consistency with numerical analysis results, positioning \cite{tuo2014surrogate}'s model as a special case within this more comprehensive framework.

Second, we introduce an innovative active learning (AL) framework. In contrast to the batch designs used by \cite{shaowu2023design} and \cite{sung2022stacking}, active learning, also known as sequential design, selects new design points one at a time based on optimizing a criterion that evaluates the model's predictive performance. This iterative, step-by-step method is not only more practical but often proves more effective than static batch designs \citep{gramacy2020surrogates}, while also taking full advantage of the underlying structure of the output expressed through our adaptive kernel. We employ the integrated mean squared prediction error (IMSPE) as the criterion for AL, taking simulation costs into account. Inspired by \cite{binois2019replication}, we show that solving this sequential design problem efficiently involves deriving a closed-form expression for IMSPE and its derivatives, enabling fast numerical optimization. In addition, we explore the challenge of selecting an adequate initial design for the active learning approach, already explored by \cite{song2025efficient} in the context of dimension reduction, highlighting the need for a balance between space-filling properties and robust parameter estimation.

Recent advances in finite element methods have introduced the statistical finite element method (statFEM), a Bayesian framework for uncertainty quantification in numerical solutions (see, e.g., \cite{papandreou2023theoretical, girolami2021statistical, akyildiz2022statistical}). While statFEM employs a hierarchical model to probabilistically account for discretization errors and model misspecification, our focus is on developing an efficient surrogate modeling and active learning framework for simulations with varying mesh configurations, offering a computationally inexpensive alternative to costly finite element simulations.

The rest of the paper is organized as follows: Section 2 introduces the proposed emulator with an adaptive non-stationary kernel. Section 3 outlines the active learning approach, detailing the  IMSPE with a focus on sequential applications, computational improvements, and initial designs. Sections 4 and 5 present numerical and real data studies, respectively. Finally, the paper concludes in Section 6.

\section{Multi-Fidelity Modeling with Continuous Fidelity Parameters}

We first present the surrogate model used in this paper in Section \ref{sec:surrogate_model}, and we introduce our novel kernel function adapted to the FEM simulations in Section \ref{sec:LBM_kernel}. We finally give additional details on inference in Sections \ref{sec:inference}.

\subsection{Non-Stationary GP Surrogate Model}\label{sec:surrogate_model}

Similarly as in the general framework of surrogate modeling for computer experiments, we study the computer (scalar) output $y(\mathbf{x},\mathbf{t}) \in \mathbb{R}$ as dependent on the computer input parameters $\mathbf{x} \in \mathcal{X}\subset \mathbb{R}^d$, which depict varying parameters that one wants to explore as part of the computer model. Unlike most surrogate models, our framework assumes that the computer output $y(\mathbf{x},\mathbf{t})$ also depends on a continuous fidelity parameter $\mathbf{t} \in \mathcal{T} \subset \mathbb{R}^m$, which controls the accuracy of the output. In the case of finite element methods, the fidelity parameter $\mathbf{t}$ would be a vector of global discretization parameters that control the resolution of the finite element solver, where smaller values of $\mathbf{t}$ correspond to finer discretizations. The input $\mathbf{x}$ represents model parameters of the physical system (not spatial locations of the finite element mesh) and is independent of $\mathbf{t}$.
The limiting case $\mathbf{t}=\mathbf{0}$ corresponds to the idealized exact solution with zero discretization error. 

We adopt the general non-stationary model introduced in \cite{tuo2014surrogate}. The approximate solution from the computer simulation $y(\mathbf{x},\mathbf{t})$ is represented as the sum of the \textit{exact/true solution} to this physical model, denoted by $\varphi(\mathbf{x}):=y(\mathbf{x},\mathbf{0})$, and the error function $\delta(\mathbf{x}, \mathbf{t})$ with respect to the fidelity parameter $\mathbf{t}$:
\begin{equation}
    y(\mathbf{x},\mathbf{t}) = \varphi(\mathbf{x}) + \delta(\mathbf{x},\mathbf{t}),~~~ (\mathbf{x},\mathbf{t}) \in \mathcal{X}\times \mathcal{T}.
\end{equation}

The unknown deterministic functions $\varphi(\mathbf{x})$ and $\delta(\mathbf{x},\mathbf{t})$ 
are modeled using Gaussian process (GP) priors under a Bayesian framework, a probabilistic framework often used to approximate expensive computer simulations  \citep{gramacy2020surrogates}. 
A GP $Y(\mathbf{x})$ is uniquely determined by a mean function $\mu(\mathbf{x})$ and a positive-definite covariance function $K(\mathbf{x}_1,\mathbf{x}_2)$. More specifically, a GP is said to be stationary if $\mu$ is constant, and $K$ can be expressed as a function of the difference between $\mathbf{x}_1$ and $\mathbf{x}_2$, i.e., $K(\mathbf{x}_1,\mathbf{x}_2) = \sigma^2 R(\mathbf{x}_1-\mathbf{x}_2)$, where $\sigma^2>0$ is the variance and $R$ is the corresponding correlation function. For instance, the Gaussian separable correlation function is $\displaystyle R_{\boldsymbol{\phi}}(\mathbf{h}) = \exp\Big(-\sum^d_{i=1}\phi_i^2 h_i^2\Big)$, where $\boldsymbol{\phi}^2 = (\phi^2_1,\ldots,\phi^2_d)$ with $\phi_i>0$ is the vector of correlation parameters along each dimension. The Mat\'ern family of kernels \citep{stein1999interpolation} is another commonly used choice, where a smoothness parameter $\nu$ controls the differentiability of the sample paths. The Gaussian kernel arises as the limiting case when $\nu \to \infty$. In practice, half-integer values of $\nu$ are often used for computational simplicity. For example, when $\nu=1.5$, the Mat\'ern correlation function is
$R_{\boldsymbol{\phi}}(\mathbf{h})=\prod^d_{i=1}\Big( 1+\sqrt{3}\phi_i|h_i| \Big) \mathrm{e}^{-\sqrt{3}\phi_i|h_i|}.$ While the framework is developed from a Bayesian perspective through GP priors, inference is conducted using plug-in estimates for the hyper-parameters, such as $\sigma^2$ and $\phi_i$. Details of the estimation procedure are provided in Section~\ref{sec:inference}.

\par 

We assume that the true solution $\varphi(\mathbf{x})$ is a realization of a non-stationary GP with a mean function $\mu_\varphi(\mathbf{x})  = \mathbf{f}_\varphi(\mathbf{x})^T\boldsymbol{\beta}_1$, where $\mathbf{f}_\varphi(\mathbf{x})$ are known regression functions, and covariance function 
\begin{equation*}
    K_\varphi(\mathbf{x}_1, \mathbf{x}_2) = \sigma^2R_{\boldsymbol{\phi}_1}(\mathbf{x}_1-\mathbf{x}_2).
\end{equation*}
Similarly, we assume that the error function $\delta(\mathbf{x},\mathbf{t})$ is a realization of a non-stationary GP with a mean function $\mu_\delta(\mathbf{x}, \mathbf{t}) = \mathbf{f}_\delta(\mathbf{x},\mathbf{t})^T\boldsymbol{\beta}_2$, where $\mathbf{f}_\delta(\mathbf{x},\mathbf{t})$ are known regression functions, and a separable covariance function in regards to the input parameter and fidelity parameter:
 \begin{equation*}
     K_\delta((\mathbf{x}_1,\mathbf{t}_1),(\mathbf{x}_2,\mathbf{t}_2)) = \sigma^2R_{\boldsymbol{\phi}_2}(\mathbf{x}_1-\mathbf{x}_2) K_{\gamma}(\mathbf{t}_1,\mathbf{t}_2), 
 \end{equation*} 
where $K_{\gamma}(\mathbf{t}_1,\mathbf{t}_2)$ defines the dependency structure between different fidelity parameter values.

The response variable $y(\mathbf{x},\mathbf{t})$ is then a realization of a GP $\{Y(\mathbf{x},\mathbf{t}):(\mathbf{x}, \mathbf{t}) \in  \mathcal{X} \times \mathcal{T}\}$ with the mean $\mu(\mathbf{x}) = \mu_\varphi(\mathbf{x}) + \mu_\delta(\mathbf{x},\mathbf{t}) = \mathbf{f}(\mathbf{x}, \mathbf{t})^T\boldsymbol{\beta}$, where $\mathbf{f}(\mathbf{x}, \mathbf{t})^T = (\mathbf{f}_\varphi(\mathbf{x})^T, \mathbf{f}_\delta(\mathbf{x}, \mathbf{t})^T)$ and  $\boldsymbol{\beta} = (\boldsymbol{\beta}_1^T, \boldsymbol{\beta}_2^T)^T \in \mathbb{R}^p$, and the covariance function
\begin{equation}\label{eq:coveq}
    K((\mathbf{x}_1,\mathbf{t}_1),(\mathbf{x}_2,\mathbf{t}_2)) = \sigma^2  \Big(R_{\boldsymbol{\phi}_1}(\mathbf{x}_1-\mathbf{x}_2) + R_{\boldsymbol{\phi}_2}(\mathbf{x}_1-\mathbf{x}_2)K_{\gamma}(\mathbf{t}_1,\mathbf{t}_2)\Big) .
\end{equation}

 As outlined in \cite{tuo2014surrogate}, the convergence of FEM methods are well studied and the error function $\delta$ can be controlled. In particular, if we consider mesh size as a single fidelity parameter and with some assumptions, then the error $\delta$ can be bounded by 
 \begin{equation}\label{eq:L2error}
     \Vert \delta \Vert_{L_2} \leq c ~ t^2 \Vert \varphi ''\Vert_{L_2},
 \end{equation}
 where $t$ is the mesh size and $c$ is a constant independent of $t$, and $\Vert \varphi ''\Vert_{L_2} = \Big(\sum_{i,j} \Vert \frac{\partial^2\varphi}{\partial x_i \partial x_j}\Vert \Big)^{\frac12}$ \citep{brenner2007fem}. This convergence of the error function can be leveraged through inducing non-stationarity in the surrogate model, by including some assumptions to ensure that the output of the computer simulation gets closer to the exact solution at a certain rate as $t \to 0$. We therefore need  that $\mu_\delta(\mathbf{x}, t) \to 0$ as $t \to 0$ and $K_{\gamma}(t_1, t_2) \to 0$ as $\min(t_1, t_2) \to 0$, ideally with a convergence rate that can be controlled. Specifically, for an error function whose convergence rate is $r$, we need that $Var(\delta(\mathbf{x},t)) = \mathcal{O}(t^l)$, where $l = 2r$ is directly determined by the convergence rate $r$. The parameter $l$ can be chosen based on upper bounds for $\delta$. For example, \eqref{eq:L2error} implies $r=2$ and thus $l=4$ for the $L_2$ norm. Further details can be found in \cite{tuo2014surrogate}.

 When considering the case of a multi-dimensional $\mathbf{t}$, this  constraint still needs to be taken into account. Moreover, as noted by \citet{ji2022multi}, we also need to ensure that if any fidelity parameter $t_j$ of $\mathbf{t}$ is non-zero, then $\lim_{\mathbf{t}_{-j}\to \mathbf{0}} P(\delta(\mathbf{x}, \mathbf{t}) = 0) = 0$, meaning that if any fidelity parameter is non-zero, then the simulation error must always have a non-negligible error, even as the other fidelity parameters tend to zero. 

From these constraints, we deduce that our model must have a non-stationary covariance function that reflects the convergence of the error function, which will be induced through the kernel function $K_{\gamma}$, introduced in the next subsection.   

\subsection{Novel Adaptive Kernel Function}\label{sec:LBM_kernel}

Given the necessity of providing the error function $\delta$ an appropriate non-stationary kernel function, we introduce a novel adaptive kernel function, aimed at capturing the behavior of the output along the fidelity parameter space. 
Specifically, we propose to use
\begin{equation} \label{eq:LB_kernel}
    K_{\gamma}(\mathbf{t}_1,\mathbf{t}_2)= \frac12 \left[ \left(\sum_{j=1}^m (a_j t_{1,j}^{l_j})^{\frac{1}{\gamma}}\right)^\gamma + \left(\sum_{j=1}^m (a_j t_{2,j}^{l_j})^{\frac{1}{\gamma}}\right)^\gamma - \left(\sum_{j=1}^m a_j^{\frac{1}{\gamma}}\left(t_{1,j}^{\frac{l_j}{2\gamma}}-t_{2,j}^{\frac{l_j}{2\gamma}}\right)^2\right)^\gamma \right],
\end{equation}
where $\gamma \in (0,1)$ controls the short range behavior of the response, $\mathbf{a} = (a_1, \cdots, a_m)$ is the scale associated with each fidelity parameter, and $\mathbf{l} = (l_1, \cdots, l_m)$ is a predefined parameter governing the convergence rate of the  error $\delta$ as described in the previous subsection. This kernel is an extension of the Lifted Brownian (LB) kriging model, introduced by \cite{plumlee2017lifted}, which defines a positive definite covariance function and is able to capture a diversity of raggedness levels through the parameter $\gamma \in (0,1)$,  thereby controlling the correlation between increments. This ability is particularly important in the study of computer experiments, especially when one is interested in the behavior of the computer output across different fidelity parameters.  We can observe in Figure \ref{fig:inc_plot} that the behavior of the output in the FEM simulations of Section \ref{sec:casestudy} can change quite dramatically depending on the mesh geometry and on the function used to define the response of interest. The resulting successive increments of the response of interest, defined as $y(t+h) - y(t)$ for a certain lag $h$, appear negatively correlated for the jet blade case study (left panel), while in the Poisson equation case study they appear somewhat uncorrelated or negatively correlated for the ``Maximum'' response (middle panel) and  positively correlated  for the ``Average'' response (right panel). This variety of sample path behavior along the fidelity parameter space gives ground for the adaptive kernel introduce in this work.

\begin{figure}[h]
    \centering
    \includegraphics[width = \linewidth]{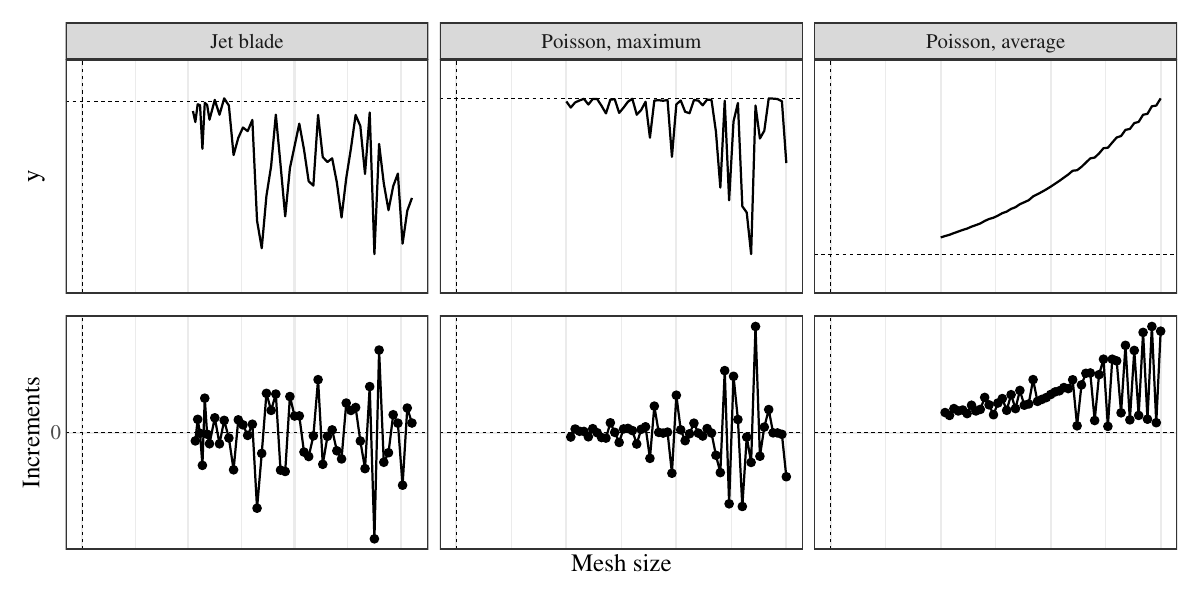}
    \caption{Results of FEM simulations at the one input point for a sequence of mesh sizes on the Poisson (middle and right) and jet blade (left panel) case studies; the 3 columns correspond to the 3 different case study problems at hand. The top row shows the response of interests $y$ depending on the mesh size with the horizontal dot line showing the convergence value and the vertical dot line at mesh size $t = 0$, while the bottom row shows the increments depending on the mesh size.}
    \label{fig:inc_plot}
\end{figure}

Previous forms of this kernel include $K_{\gamma}(t_1, t_2) = \min(t_1^l,t_2^l)$, used by \citet{tuo2014surrogate} for models with a single fidelity parameter, a special case of our kernel with $\gamma = 0.5$, referred to as Brownian motion (BM); in the case of a multi-dimensional fidelity parameter, \cite{ji2022multi} proposed $K_{\gamma}(\mathbf{t}_1, \mathbf{t}_2) = e^{-\sum_{j=1}^m\theta_j^2 (t_{1,j}-t_{2,j})^2}-e^{-\sum_{j=1}^m\theta_j^2 t_{1,j}^2}-e^{-\sum_{j=1}^m\theta_j^2 t_{2,j}^2}+1$ .

This new kernel possess several advantages, both in terms of theoretical grounding and in its ability to reproduce sample path behaviors that more closely reflect observations from computer simulations. First,  as $\mathbf{t} \to \mathbf{0}$, we have $K_{\gamma}(\mathbf{t}_1, \mathbf{t}_2) \to 0$. Moreover, the variance associated with the fidelity parameter satisfies $K_{\gamma}(\mathbf{t}, \mathbf{t}) = \Vert\mathbf{a}\odot\mathbf{t}^{\mathbf{l}}\Vert_{\frac{1}{\gamma}} \geq \max_j(a_j t_j^{l_j})$, where $\Vert . \Vert_p$ indicates the $p$-norm on vectors. This lower bound ensures that the convergence rate is consistent with the error bounds discussed in Section~\ref{sec:surrogate_model}, and further guarantees that if any component of the fidelity parameter is non-zero, then $P(\delta(\mathbf{x}, \mathbf{t}) = 0) = 0$ (see the Supplementary Materials for the formal proof). \par

\subsection{Statistical Inference}\label{sec:inference}

Suppose we have conducted $n$ computer experiments at design locations $\mathbf{X}_n = (\mathbf{x}_1^T,\cdots,\mathbf{x}_n^T)$ and fidelity parameters $\mathbf{T}_n = (\mathbf{t}_1,\cdots,\mathbf{t}_n)^T$, with the corresponding outputs $\mathbf{y}_n = (y_1,\cdots,y_n)^T$, i.e., $y_i=y(\mathbf{x}_i,\mathbf{t}_i)$.\par

The hyper-parameters $\{\boldsymbol{\beta},\sigma^2,\boldsymbol{\phi}^2_1,\boldsymbol{\phi}^2_2,\mathbf{a},\gamma\}$ are estimated using the Maximum Likelihood Estimation (MLE). First, we express the covariance function in (\ref{eq:coveq}) as $\displaystyle K((\mathbf{x}_1,\mathbf{t}_1),(\mathbf{x}_2,\mathbf{t}_2)) =  \sigma^2  K_0((\mathbf{x}_1,\mathbf{t}_1),(\mathbf{x}_2,\mathbf{t}_2))$. Using the Restricted Maximum Likelihood Estimation (RMLE), we obtain the profile likelihood for $\boldsymbol{\beta}$ and $\sigma^2$, and then estimate the remaining hyper-parameters \citep{stein1999interpolation,santner2018design}. The profile likelihood to be maximized is 
\begin{equation*}
    l(\boldsymbol{\theta}; \mathbf{y}_n) =\text{constant} - \frac{n-p}{2} \log \Big\{ \mathbf{Z}^T\Big(\mathbf{K}_{0,n}^{-1} - \mathbf{K}_{0,n}^{-1}\mathbf{F}_n\mathbf{P}_n^{-1} \mathbf{F}^T_n \mathbf{K}_{0,n}^{-1} \Big) \mathbf{Z} \Big\} 
    -\frac{1}{2} \log \vert \mathbf{K}_{0,n}\vert  -\frac{1}{2} \log \vert \mathbf{P}_n\vert
\end{equation*}
 where $\mathbf{F}_n = \big(\mathbf{f}(\mathbf{x}_i, t_i)^T)_{1\leq i \leq n}$, $\mathbf{Z} = (\mathbf{I} - \mathbf{F}_n(\mathbf{F}_n^T\mathbf{F}_n)^{-1}\mathbf{F}_n^T)\mathbf{y}_n$, $\mathbf{K}_{0,n} = \left(K_0((\mathbf{x}_i,\mathbf{t}_i),(\mathbf{x}_j,\mathbf{t}_j))\right)_{i,j}$ and $\mathbf{P}_n = \mathbf{F}_n^T \mathbf{K}_{0,n}^{-1}\mathbf{F}_n$, and $p$ is the number of components in $\boldsymbol{\beta}$. The hyper-parameters $\{\boldsymbol{\phi}^2_1,\boldsymbol{\phi}^2_2,\mathbf{a},\gamma\}$ can then be estimated by maximizing the likelihood  $l(\boldsymbol{\theta}; \mathbf{y}_n)$.
The resulting MLEs for $\boldsymbol{\beta}$ and $\sigma^2$ are $\hat{\boldsymbol{\beta}} =(\mathbf{F}_n^T\mathbf{K}_{0,n}^{-1}\mathbf{F}_n)^{-1}\mathbf{F}_n^T\mathbf{K}_{0,n}^{-1}\mathbf{y}_n$ and $ \hat{\sigma}^2 = \frac{(\mathbf{y}_n-\mathbf{F}_n\hat{\boldsymbol{\beta}})^T\mathbf{K}_{0,n}^{-1}(\mathbf{y}_n-\mathbf{F}_n\hat{\boldsymbol{\beta}})}{n-p}$, where $\mathbf{K}_{0,n}$ is calculated with the plug-in estimates $\{\hat{\boldsymbol{\phi}}^2_1,\hat{\boldsymbol{\phi}}^2_2,\hat{\mathbf{a}},\hat{\gamma}\}$.\par

The maximization of the log-likelihood can be performed using a quasi-Newton optimization method \citep{byrd1995limited} using the closed-form expressions for the gradient of the log-likelihood, which we have included in the Supplementary Materials.
Compared to the fully Bayesian approach in \cite{tuo2014surrogate}, our method is more computationally efficient and avoids the need for informed priors.

By the property of conditional normal distributions, the posterior distribution of $y(\mathbf{x},\mathbf{t})$ given the observations $\mathbf{y}_n$ follows a normal distribution with predictive mean and variance:
\begin{align}
    \mu_n(\mathbf{x},\mathbf{t}) & = \mathbf{f}(\mathbf{x},\mathbf{t})^T \boldsymbol{\beta} + \mathbf{k}_n(\mathbf{x},\mathbf{t})^T \mathbf{K}_n^{-1}(\mathbf{y}_n - \mathbf{F}_n \boldsymbol{\beta}),  \nonumber\\  
    \sigma_n^2(\mathbf{x},\mathbf{t}) & = K(\mathbf{x},\mathbf{x},\mathbf{t},\mathbf{t}) - \mathbf{k}_n(\mathbf{x},\mathbf{t})^T \mathbf{K}_n^{-1} \mathbf{k}_n(\mathbf{x},\mathbf{t}) + \boldsymbol{\gamma}_n(\mathbf{x},\mathbf{t})^T(\mathbf{F}_n^T\mathbf{K}_n^{-1}\mathbf{F}_n)^{-1}\boldsymbol{\gamma}_n(\mathbf{x},\mathbf{t}), \label{eq:pred_var}
\end{align}
where $\mathbf{k}_n(\mathbf{x},\mathbf{t})^T = \big(K((\mathbf{x},\mathbf{t}),(\mathbf{x}_i,\mathbf{t}_i))\big)_{1\leq i \leq n}$, $\mathbf{K}_n = \big(K((\mathbf{x}_i,\mathbf{t}_i),(\mathbf{x}_j,\mathbf{t}_j))\big)_{i,j}$, and $\boldsymbol{\gamma}_n(\mathbf{x},\mathbf{t}) = \big(\mathbf{f}(\mathbf{x},\mathbf{t}) - \mathbf{F}_n^T\mathbf{K}_n^{-1}\mathbf{k}_n(\mathbf{x},\mathbf{t})\big)$. The  hyper-parameters in the posterior distribution can be plugged in by their  estimates.

From kriging theory \citep{cressie2015statistics}, this posterior predictive distribution is the Best Linear Unbiased Prediction (BLUP) for $y(\mathbf{x}, \mathbf{t})$. Unlike other multi-fidelity surrogate models that rely on auto-correlation models \citep{kennedy2000predicting, le2015cokriging, heo2023active}, our surrogate model is capable of extrapolation, making predictions for the true solution at $\mathbf{t} = \mathbf{0}$, while providing uncertainty quantification with the predictions.
 
\section{Active Learning}\label{sec:active_learning}

In the context of computer experiments, where the need to carefully manage limited computational resources intersects with the flexibility to choose design points without constraints, active learning methods have naturally gained popularity, particularly when combined with GP surrogate models \citep{rasmussen2006gaussian, santner2018design,gramacy2020surrogates}. In this work, we focus specifically on the objective of minimizing predictive error of the true solution. 


\subsection{Cost Adjusted IMSPE Reduction}\label{sec:criterion}

We employ the Integrated Mean Squared Prediction Error (IMSPE) \citep{gramacy2020surrogates} as the foundation of our active learning criterion. 
IMSPE provides an intuitive criterion for sequentially selecting design points that directly targets global predictive accuracy of the GP surrogate. Specifically, the IMSPE measures the total posterior uncertainty in our prediction by integrating the Mean Squared Prediction Error (MSPE) across the entire domain space. Our aim being to predict the true solution $y(\mathbf{x}, \mathbf{0})$, we integrate the MSPE associated with the posterior prediction $\mu_n(\mathbf{x}, \mathbf{0})$ over the input space $\mathcal{X}$ only, considering that the prediction over the rest of the domain $\mathcal{T}$ is of little interest. Since our model specification implies that the prediction is unbiased, the MSPE is simply equal to the predictive variance $\sigma_n^2(\mathbf{x},\mathbf{0})$ from (\ref{eq:pred_var}). Given this, we define the IMSPE as
\begin{equation}
    \mathrm{IMSPE}(\mathbf{X}_n, \mathbf{T}_n) = \int_{\mathbf{x}\in \mathcal{X}} \sigma_n^2(\mathbf{x},\mathbf{0}) \text{d}\mathbf{x} := I_n.
\end{equation}

Given the current design $(\mathbf{X}_n, \mathbf{T}_n)$, our active learning objective is to find the next best design location  $(\mathbf{x}_{n+1}, \mathbf{t}_{n+1})$ by minimizing $I_{n+1}(\mathbf{x}_{n+1}, \mathbf{t}_{n+1}) := \mathrm{IMSPE}(\mathbf{X}_{n+1}, \mathbf{T}_{n+1})$, where $\mathbf{X}_{n+1} = \{\mathbf{x}_1, \cdots,\mathbf{x}_n,\mathbf{x}_{n+1}\}$ and $\mathbf{T}_{n+1} = \{\mathbf{t}_1, \cdots, \mathbf{t}_n, \mathbf{t}_{n+1}\}$ are the combined set of $n$ current design points together with the new input location for the input parameter and fidelity parameter. For any candidate point $(\tilde{\mathbf{x}}, \tilde{\mathbf{t}})$, the following theorem shows that the IMSPE $I_{n+1}(\tilde{\mathbf{x}}, \tilde{\mathbf{t}})$, can be written in a sequential manner, facilitating efficient computation when applied to an active learning method.

\begin{theorem}\label{theorem:imspe_seq}
The IMSPE associated with an additional design point $(\tilde{\mathbf{x}}, \tilde{\mathbf{t}})$ given the current design $(\mathbf{X}_n, \mathbf{T}_n)$ can be written in an iterative form as
    \begin{equation} I_{n+1}(\tilde{\mathbf{x}}, \tilde{\mathbf{t}}) = I_n - R_{n+1}(\tilde{\mathbf{x}}, \tilde{\mathbf{t}}),\end{equation}
    where $R_{n+1}(\tilde{\mathbf{x}}, \tilde{\mathbf{t}})$, the IMSPE reduction, has a closed-form expression under both the Gaussian and Mat\'ern correlation functions, with a computational cost of $\mathcal{O}(n^2)$.
\end{theorem}

\begin{proof}
    The proof, along with the closed-form expression of $R_{n+1}(\tilde{\mathbf{x}}, \tilde{\mathbf{t}})$  for the Gaussian and Matérn correlation functions with $\nu=0.5$, $\nu=1.5$, and $\nu=2.5$, can be found in the Supplementary Materials.
\end{proof}

A straightforward observation is that minimizing the IMSPE with an additional design point $(\tilde{\mathbf{x}}, \tilde{\mathbf{t}})$ is equivalent to maximizing the IMSPE reduction, $R_{n+1}(\tilde{\mathbf{x}}, \tilde{\mathbf{t}})$. We will, therefore, use $R_{n+1}(\tilde{\mathbf{x}}, \tilde{\mathbf{t}})$ as the basis for our active learning criterion. \par

Because FEM simulations require managing the trade-off between accuracy and computational cost, our active learning method would be biased towards high-fidelity data points if we didn't account for the computational cost. Considering the computational cost $C(\mathbf{t})$, our criterion selects the next design point  $(\mathbf{x}_{n+1}, \mathbf{t}_{n+1})$ as
\begin{equation}\label{eq:costawareobj}
    (\mathbf{x}_{n+1}, \mathbf{t}_{n+1}) = \argmax_{(\tilde{\mathbf{x}}, \tilde{\mathbf{t})} \in \mathcal{X} \times \mathcal{T}} \frac{R_{n+1}{(\tilde{\mathbf{x}}, \tilde{\mathbf{t}})}}{C(\tilde{\mathbf{t}})}.
\end{equation}

An alternative formulation is to maximize $R_{n+1}(\tilde{\mathbf{x}}, \tilde{\mathbf{t}})$ subject to a total computational budget constraint. Introducing this constraint via a Lagrangian leads to an objective that reflects the same trade-off between information gain and computational cost through a cost-penalty parameter (i.e., the Lagrange multiplier). However, selecting this cost-penalty parameter can be nontrivial in practice. Following previous works \citep{stroh2022sequential, he2017optimization} that implement a ratio-based criterion, we use the ratio between the IMSPE reduction and the computational cost, as presented in Equation \eqref{eq:costawareobj}.

Taking advantage of the continuous aspect of the domain on which the optimization is performed and of the fast computation of the criterion offered by the closed-form expression given in Theorem \ref{theorem:imspe_seq}, we can use library-based numerical schemes to optimize our criterion. This is a novelty to the best of our knowledge, as other works have only explored active learning methods on a small discrete set of pre-defined fidelity levels \citep{stroh2022sequential, sung2022stacking, he2017optimization}, overlooking design points with higher potential as shown in Figure \ref{fig:illustration_AL}. In addition, these methods require to evaluate candidate design points on an exhaustive regular grid using computational-heavy methods such as Monte-Carlo approximation to evaluate the criterion. An \textsf{R} package \textsf{MuFiMeshGP} for the active learning method is available in an open repository and the Supplementary Materials.

\begin{figure}[h]
    \centering
    \includegraphics[width=\linewidth]{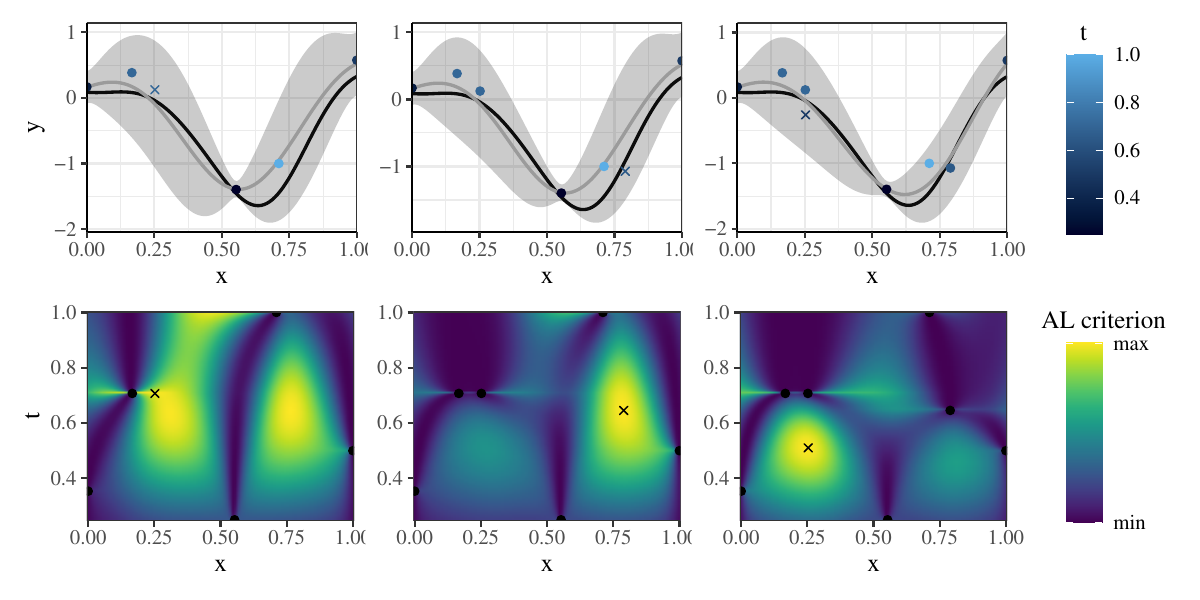}
    \caption{Top: Prediction of our model with the true function (black line), the prediction mean (gray line), and the 95\% prediction band (shaded region). The color of the design points indicate the corresponding fidelity parameter, with a darker color indicating a lower value. Bottom: Active learning criterion surface (bottom). In both panels, the bullet points $(\bullet)$ represent the current design locations, while the crosses $(\times)$ indicate the best next design point according to the criterion.}
    \label{fig:illustration_AL}
\end{figure}

Figure \ref{fig:illustration_AL} illustrates the iterative process of selecting the optimal design point based on our active learning criterion. Starting from an initial design (left panel), we fit our model and optimize our criterion to obtain the next best design point (cross dot). A computer simulation is then run at this additional design point and added to our current design (middle panel), after which we update the model. This process continues until the time budget is exhausted or a stopping criterion is met. Notably, the criterion is non-convex, requiring a multi-start optimization approach with a gradient-descent algorithm.

The underlying function in Figure \ref{fig:illustration_AL} is taken from the simulation study in Section \ref{sec:numerical_study}, and represents a case where the increments are uncorrelated (i.e., $\gamma = 0.5$ in our model). Our active learning  leverages this information by favoring data points that are at higher fidelity in the regions where the input space has already been explored (see the right panel). Indeed, since the increments along the fidelity parameter space are uncorrelated, the posterior at an already explored input point depends only on the observation with the highest available fidelity at that location. In this case, after selecting a couple sample points that help explore the input space, the active learning criterion favors data points with lower fidelity parameter values, indicating a preference for higher fidelity simulations. 

The horizontal streaks of lighter color in the active learning criterion (bottom of Figure \ref{fig:illustration_AL})  reflect that our method also accounts for the correlation in the input space: sampling  neighboring inputs at the fidelity levels that have already been explored can be  beneficial. 
Additional examples for $\gamma = 0.05$ and $\gamma = 0.95$ are provided in Figures~S1 and~S2 of the Supplementary Materials. In summary, we observe that it displays different patterns depending on the hyper-parameter $\gamma$, highlighting its adaptability and responsiveness. 

\subsection{Initial Designs for Active Learning}\label{sec:init_design}

In the previous subsection, we highlight the ability of our active learning method to be responsive to the hyper-parameters of the model, suggesting that a robust estimation of the hyper-parameters from the initial design might be advantageous for the subsequent active learning process. 
This issue was also raised by \cite{song2025efficient} in the context of screening designs for GP surrogate models, where they argue that space-filling designs are not necessarily the most optimal when used as initial designs for active learning.  Indeed, while space-filling designs and active learning criteria are typically aimed at improving predictive performance, they often do not explicitly account for uncertainty in parameter estimation, leading to unstable or unreliable hyperparameter estimates, which in turn may affect the robustness of subsequent model updates. Despite the extensive literature on active learning, the problem of optimizing the initial design to improve overall predictive performance has been mostly reduced to the choice of initial sample size required to obtain good performances \citep{loeppky2009choosing, harari2018computers}.

Our parameter estimation problem is driven by two main sources of uncertainties, both of which we believe can be addressed through better initial designs. The first arises from the inherent challenge in multi-fidelity simulations: the separation between the behavior of the real solution and the error function. In modeling terms, we need robust estimation of the parameter $\boldsymbol{\phi}^2_2$ that is able to differentiate it from $\boldsymbol{\phi}^2_1$. Additionally, the parameter $\gamma$ from the LBM kernel significantly influences the behavior of our active learning method, making a good prior estimation crucial. These factors lead us to consider designs with a stacked structure to better capture the correlation in increments along the fidelity parameter, thereby more accurately estimating  $\gamma$, as well as a nested structure to better understand the behavior of the error function and improve the estimation of $\boldsymbol{\phi}^2_2$. These designs will be compared to designs that have good projection and space-filling properties such as the Maximum Projection design (\texttt{MaxPro}, \cite{joseph2015maximum}), and the Multi-Mesh Experimental Design (\texttt{MMED}, \cite{shaowu2023design}), an extension of the \texttt{MaxPro} design in which the fidelity parameter values are chosen based on a weighted maximin criterion that accounts for computational cost. Further discussions are included in the Supplementary Materials, where we examine the impact of the initial design on the robustness of the hyper-parameter estimation and on the subsequent performance of the active learning procedure. These studies ultimately indicate that while the choice of initial design has a substantial effect on the robustness of hyperparameter estimation, the prediction performance of our active learning method is not significantly impacted by the choice of initial design.

\section{Numerical Study}\label{sec:numerical_study}

In this section, we perform numerical experiments within the framework of the multi-fidelity model to assess the performance of our active learning method with various initial designs and different non-stationary kernels, and compare these results against predictions made from a single-fidelity surrogate model.

We generate 10 sets of samples from our multi-fidelity GP surrogate model for each parameter combination $(\phi_2^2,\gamma)$, with $\phi_2^2 \in \{1, 10, 100\}$ and $\gamma \in \{0.05, 0.5, 0.95\}$, while fixing the remaining parameters at $\phi_1^2 = 1$, $a = 1$, $l = 4$, and $\sigma^2 = 1$. Each set consists of five repetitions to ensure a sufficiently large sample. We compare the performance of the \texttt{LBM} kernel with the \texttt{BM} kernel and assess one-shot (\texttt{OS}) designs against our active learning approach (\texttt{AL}). We also include a single fidelity GP surrogate model at fidelity parameter $t = 0.25$ (\texttt{SF}) to benchmark performance against the highest available fidelity level. The fidelity parameter is restricted to the range $[0.25, 1]$, with a total simulation budget of 128 and a cost function given by $C(t) = t^{-2}$. The active learning procedure is performed with an initial budget of 64 allocated towards the initial design.  The \texttt{MMED} design is used for the one-shot design and the initial design in the active learning method. 

\begin{figure}[h]
    \centering
    \includegraphics[width=\linewidth]{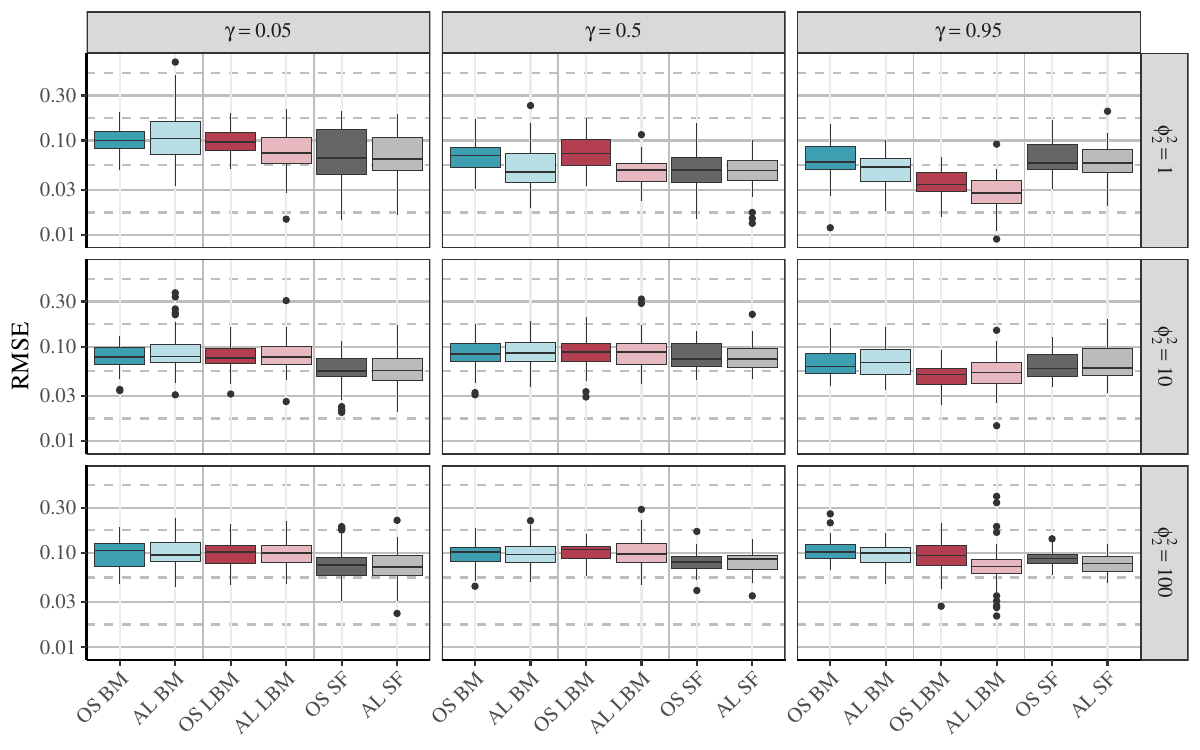}
    \caption{RMSE (in logarithmic scale) for the numerical study for each method, and all combinations of $\phi_2^2$ and $\gamma$.}
    \label{fig:sim_boxplot}
\end{figure}


The resulting predictive performance in terms of Root Mean Squared Error (RMSE) for the different methods is summarized in Figure \ref{fig:sim_boxplot}. As expected, the single-fidelity surrogate models perform comparatively better when the underlying process has a more complex behavior on the fidelity parameter space (lower $\gamma$) and when the error function displays sharp fluctuations (higher $\phi_2^2$). Despite being deliberately chosen to be highly competitive, the single-fidelity model is outperformed by our approach in favorable settings ($\gamma = 0.95)$, while our method remains highly competitive across all other scenarios. We also observe that the proposed \texttt{LBM} kernel performs considerably better than its \texttt{BM} counterpart when increments along the fidelity parameter are positively correlated ($\gamma = 0.95$), reflecting its ability  to capture this dependence structure. Finally, the active learning method applied on the multi-fidelity surrogate models shows its effectiveness as it seldom performs worse than its one-shot counterpart, and yields superior predictive performance in many cases. 

\section{Case Studies}\label{sec:casestudy}

We assess the performance of our proposed surrogate model and active learning method through three case studies. Similar to the setup of Section \ref{sec:numerical_study}, we  compare the \texttt{LBM} kernel with the \texttt{BM} kernel, and contrast our active learning method (\texttt{AL}) with the one-shot design (\texttt{OS}). The \texttt{MMED} design is used both as the one-shot design and as the initial design for active learning. We further benchmark our approach against the stacking design method proposed by \citet{sung2022stacking}.


To ensure fair comparisons, the cost function $C(\mathbf{t})$ for running a computer experiment is assumed to be known, and all surrogate models are allocated the same fixed computational budget. The stacking design method employs a target error as a stopping criterion rather than a budget; therefore, we select an appropriate target error so that the total computational cost matches the prescribed budget.

Choosing a fidelity parameter for single-fidelity models is nontrivial, as it requires balancing the number of design points against the fidelity of each simulation. Since it is impossible to guess a priori what is the optimal fidelity parameter, we select a small set of fidelity parameters that seem judicious given the cost function, total budget, and input dimension. For example, in the Poisson's equation case of Section \ref{sec:poisson_ex}, we consider final design sizes of $n = 6, 8, 10$. The corresponding fidelity parameter values are obtained by inverting the cost function $C$ given the total budget $b$, i.e., $t = C^{-1}\big(\frac{b}{n}\big)$.

In the case studies presented in Section \ref{sec:poisson_ex} and \ref{sec:jetblade_ex}, our mean regression functions, $\mathbf{f}_V(\mathbf{x})$, are Legendre $2^{\mathrm{nd}}$ degree polynomials on $[0,1]$, supplemented with the interaction functions $ \left((2x_i-1)(2x_j-1)\right)_{1\leq i<j\leq d}$. Additionally, we use $\mathbf{f}_Z(\mathbf{x}, t) = t^l$ to represent a possible trend effect in the underlying FEM simulations.  For the high-dimensional case study in Section~\ref{sec:wave_ex}, we opt to use a constant mean regression function $\mathbf{f}(\mathbf{x}, \mathbf{t}) = \mathbf{1}$.

Sections~\ref{sec:poisson_ex} and~\ref{sec:jetblade_ex} consider mesh size as the sole fidelity parameter; accordingly, we adopt the bound on $\delta$ given in~\eqref{eq:L2error}, which implies $l=4$.  Section~\ref{sec:wave_ex} involves two fidelity parameters: mesh size and time-step size, so $\mathbf{l}=(l_1,l_2)$ is two-dimensional. For the mesh dimension, we again set $l_1=4$. For the time-step dimension, due to the lack of explicit convergence rate information for the numerical solver, we assume the weakest plausible rate $r=1$, leading to $l_2=2$.

All the computer experiments were performed using finite element methods (FEM), specifically leveraging the PDE toolbox of \cite{matlabpde}. In all the case studies, the elements that constitute the mesh are controlled in size by the fidelity parameter, a scalar, that corresponds to the maximum size of the element's edges. In the case of a time-dependent PDE, the fidelity parameter vector $\mathbf{t}$ additionally includes a component governing the time-step size.

\subsection{Poisson's Equation}\label{sec:poisson_ex}

We first use an elliptic PDE system to assess the performance of our method. The system of interest is modeled using Poisson's equation on the square membrane $D = [0,1]\times [0,1]$, 
\begin{equation*}
\Delta u = (x^2 - 2\pi^2)e^{xz_1} \sin(\pi z_1) \sin(\pi z_2) + 2\pi x e^{xz_1} \cos(\pi z_1) \sin(\pi z_2), ~~~~~~(z_1, z_2) \in D,
\end{equation*}
where $u(z_1, z_2)$ is the solution of interest, $\Delta  = \frac{\partial^2}{\partial z_1^2} + \frac{\partial^2}{\partial z_2^2}$ is the Laplace operator, and $x \in  \mathcal{X}  = [-1, 1]$ is our input parameter. A Dirichlet boundary condition $u = 0$ on the boundary $\partial D$ is imposed on the system.

We consider two different responses of interest for this study: the average over $D$ and the maximum over $D$. The closed-form solution of this system given $x$ is $\displaystyle u_x(z_1, z_2) = e^{xz_1} \sin(\pi z_1)\sin(\pi z_2)$, allowing us to derive an analytical expression for its average over $D$: $\displaystyle \varphi(x) = \frac{2(e^x+1)}{x^2+\pi^2}$. For the maximum of the true solution over  $D$, we determine it numerically by optimizing $u_x$ over $D$: $\displaystyle \varphi(x) = \max_{(z_1 , z_2)\in D} u_x(z_1, z_2)$. The cost function is taken as $C(t) = t^{-2}$.

The numerical experiments are repeated 10 times for each method with a total budget of 1500 and a fidelity parameter space ranging between $\mathcal{T} =  [0.25/\sqrt{10},0.25]$. For  active learning, we allocate a budget of 500 toward the initial design of the multi-fidelity surrogate models, and use an initial design of 4 design points for the single-fidelity surrogate models. The predictive performance of the active learning methods against the cost budget is shown in Figure \ref{fig:poisson_AL}, while Figure \ref{fig:boxplot_poisson} shows the final RMSE for all methods.

\begin{figure}[h]
    \centering
    \includegraphics[width=\linewidth]{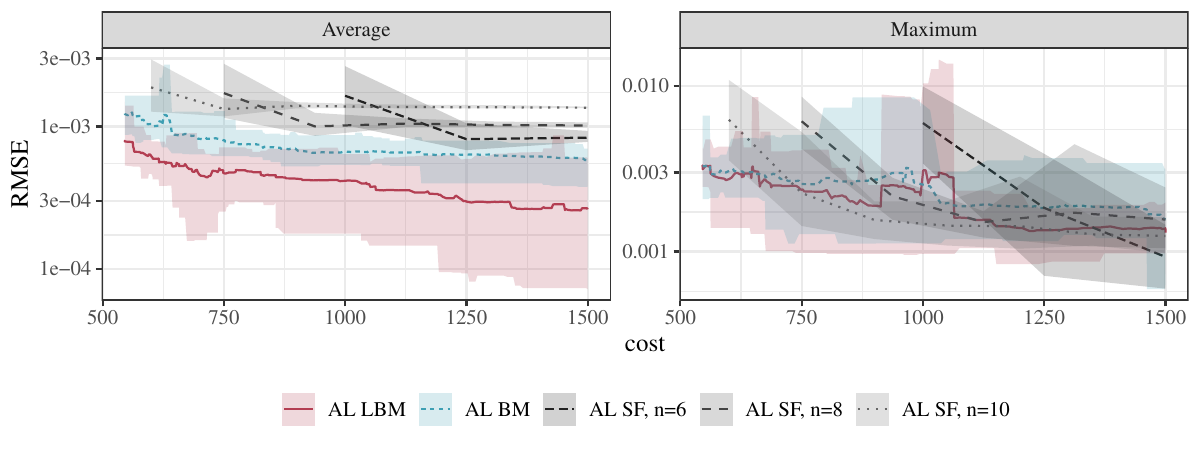}
    \caption{RMSE (in logarithmic scale) for the Poisson's equation case study with respect to the simulation cost. Solid lines indicate the average over 10 repetitions, while shaded regions represent the range. The response of interest is the average (left) or the maximum (right) over the domain $D$.}
    \label{fig:poisson_AL}
\end{figure}

\begin{figure}[h!]
    \centering
    \includegraphics[width = \linewidth]{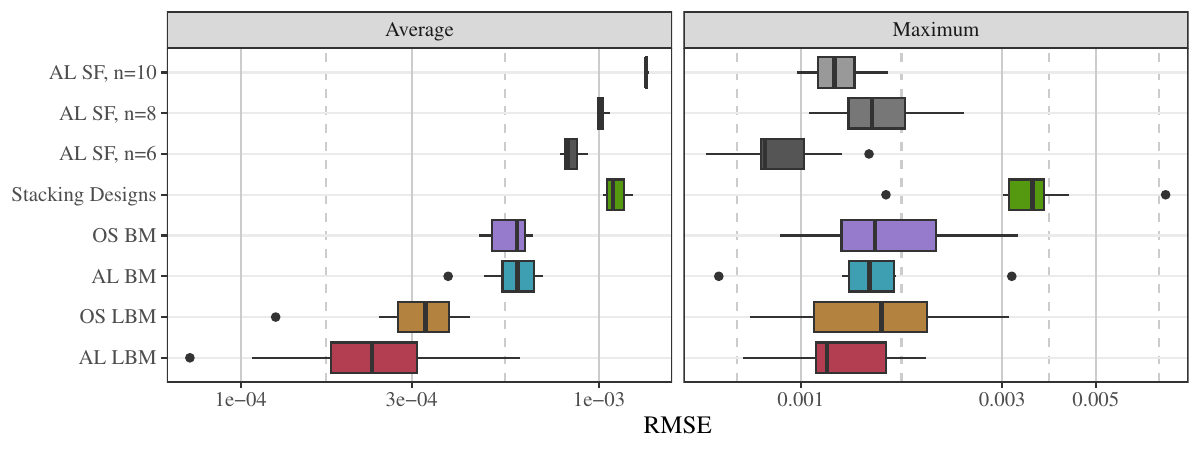}
    \caption{Boxplots of the final RMSE in logarithmic scale for the Poisson's equation case study across 10 repetitions. The response of interest is the average (left) or the maximum (right) over the domain $D$.}
    \label{fig:boxplot_poisson}
\end{figure}

Recall that Figure \ref{fig:inc_plot} illustrates the correlated structure of the two responses of interest, suggesting that our model is expected to perform particularly well for the ``Average'' response by leveraging this correlation. This is confirmed by the results, showing that the \texttt{LBM} kernel significantly outperforms other methods, while the performance differences are less pronounced for the ``Maximum'' case. Nevertheless, our model remains highly competitive in this case, especially when compared to the highest-fidelity setting ($n = 6$), and provides a significant advantage over the single fidelity model as the single fidelity model requires selecting a mesh size prior to the start of the FEM simulation. If the chosen mesh size is too coarse, predictive accuracy suffers; if it is too fine, the surrogate model may fail to converge due to budget constraints, as evident in the right panel of Figure \ref{fig:poisson_AL}. In contrast, the multi-fidelity surrogate model offers the flexibility to incorporate various mesh sizes, enabling faster convergence within a fixed budget.

Figure \ref{fig:boxplot_poisson} shows that the active learning method consistently outperforms its one-shot design counterparts, highlighting the superiority of active learning when it comes to predictive performances. This advantage is particularly pronounced when active learning is combined with the proposed \texttt{LBM} kernel, highlighting its ability to leverage the structure of the kernel. 

\subsection{Stress Analysis of a Jet Engine Turbine Blade}\label{sec:jetblade_ex}

In this second case study, we investigate the performance of our proposed model on a static structural analysis application for a jet turbine engine blade in steady-state operating condition. The turbine is a component of the jet engine, typically made from nickel alloys that resist extremely high temperatures, which needs to withstand high stress and deformations to avoid mechanical failure and friction between the tip of the blade and the turbine casing. 

The effect of thermal stress and pressure of the surrounding gases on turbine blades can be solved with FEM as a static structural model. There are two input variables: the pressure load on the pressure $(x_1)$, and suction $(x_2)$ sides of the blade, both of which range from 0.25 to 0.75 MPa, i.e., $(x_1, x_2) \in  \mathcal{X}  = [0.25, 0.75]^2$. The response of interest in this problem is the displacement of the tip of the blade towards the casing, i.e., in the y-axis direction, to determine the pressure that the turbine can withstand before friction between the blade and the casing starts to appear (see Figure \ref{fig:jetblade_illustration}).

\begin{figure}[h]
    \centering
    \includegraphics[width=0.8\linewidth]{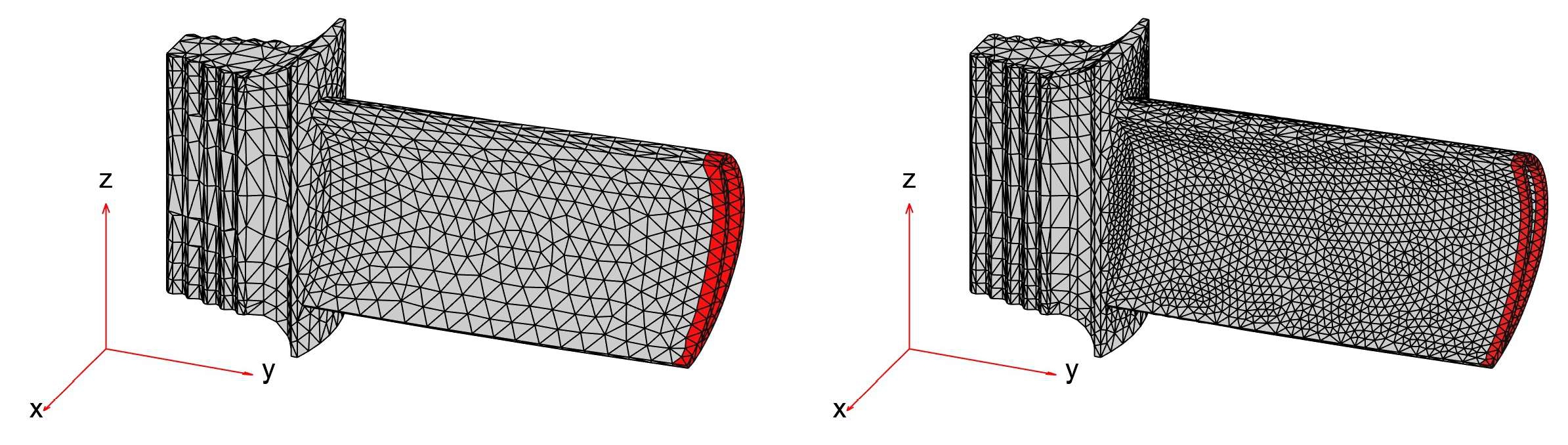}
    \caption{Visualizing the FEM mesh generation for two choices of mesh sizes ($t = 0.005$ on the left and $t = 0.0025$ on the right) in the turbine blade case study. The red cells indicate the area of interest.}
    \label{fig:jetblade_illustration}
\end{figure}

Unlike in Section \ref{sec:poisson_ex}, the true solution cannot be expressed in closed form; we thus perform validation runs at 50 points sampled from a \texttt{MaxPro} design on the input space using a mesh size of $t = 0.001$. Those simulations each took an average of 7500 seconds to complete on a High Performance Computing Center\if1\blind{ of the Institute for Cyber-Enabled Research (ICER) at MSU}\fi, utilizing up to 500GB of memory. The cost function is approximated as $C(t) = bt^{-a}$, where $a,b \in \mathbb{R}^+$ would be estimated from the initial design. For the sake of consistency, we ran a separate set of 20 simulations to estimate this function. The initial designs used by the active learning method include mesh sizes in the range $t \in [0.0025, 0.005]$.  The total budget is set to 1500, with an initial budget of 500 allocated to the active learning for multi-fidelity models, while single-fidelity models starts with an initial design of size $n = 10$.

The results of this study are summarized in Figure \ref{fig:jetblade_results}, showing that the \texttt{LBM} kernel outperforms the \texttt{BM} kernel and single-fidelity models. Notably, we observe that selecting a finer mesh size in single-fidelity models in this case does not necessarily lead to better predictive performance, demonstrating the advantage of the multi-fidelity approach: robust predictive performance without the need of fine-fidelity the design size and fidelity parameter value. Furthermore, active learning consistently improves predictions over one-shot designs, particularly with the \texttt{LBM} kernel, and outperforms stacking designs.

\begin{figure}[h]
    \centering
    \includegraphics[width = \linewidth]{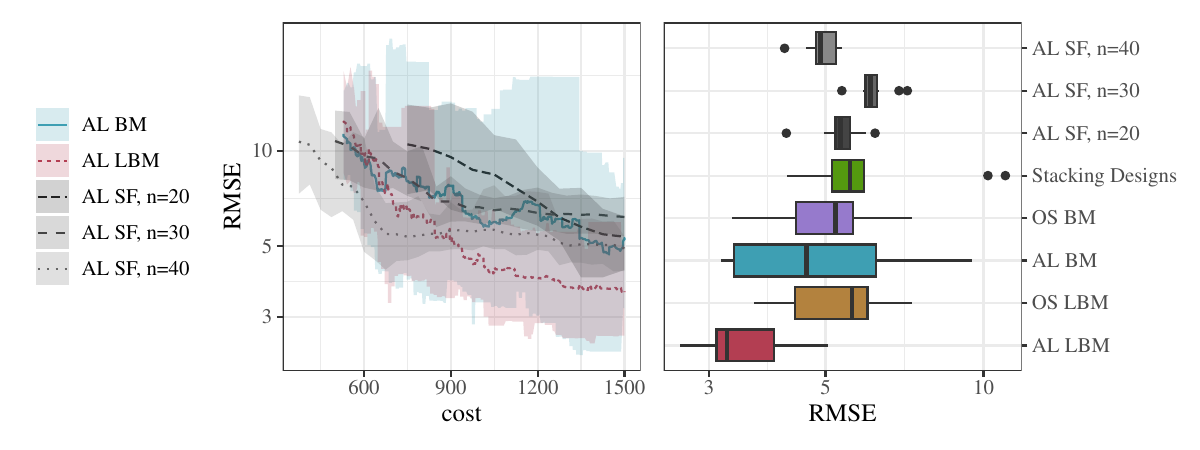}
    \caption{Left: RMSE (in logarithmic scale) for the jet engine turbine blade case study with respect to the simulation cost. Solid lines indicate the average over 10 repetitions, while shaded regions represent the range. Right: Boxplots of the final RMSE in logarithmic scale for the jet engine turbine blade study across 10 repetitions.}
    \label{fig:jetblade_results}
\end{figure}


\subsection{Wave Equation}\label{sec:wave_ex}

Finally, we study the wave equation, a time-dependent hyperbolic PDE, on a square membrane $D_z = [-1,1] \times [-1,1]$ and on the time domain $D_{\tau} = [0,5]$. The system of interest is $\Delta u = \frac{\partial^2 u}{\partial \tau^2}$, with a Neumann boundary condition $\frac{\partial u}{\partial v} = 0$ on the boundary $\partial D$ is imposed on the system, where $v$ is the outward unit normal with respect to $\partial D$. 
The initial displacement is defined as two Gaussian peaks: $u_0(z_1, z_2) = e^{-5((z_1 - x_1)^2 + (z_2 - x_2)^2))} + e^{-5((z_1 - x_3)^2 + (z_2 - x_4)^2))}$, where $(z_1, z_2) \in D_z$ and $\mathbf{x} = (x_1, x_2, x_3, x_4) \in [-1,1]^4$ is our input vector. The input vector $\mathbf{x}$ specifies the shape of the initial displacement, which is comprised of two peaks whose locations are determined by the pairs $(x_1, x_2)$ and $(x_3, x_4)$. 
In this FEM simulation, we consider two fidelity parameters, $\mathbf{t}=(t_1,t_2)$, corresponding to the mesh size and the time step size.

We divide the time domain into the initial phase $D_{\tau}^{\text{init}} = [0,1)$ and the study phase $D_{\tau}^{\text{study}} = [1,4]$ in order to let the system have enough time to transition from the initial conditions. The response of interest is then taken as the maximum of the displacement $u$ over the domain $D_z$ and the study phase $D_{\tau}^{\text{study}}$. Similarly as in Section \ref{sec:jetblade_ex}, the equation does not have an analytical solution, so we perform validation runs at 200 points sampled from a \texttt{MaxPro} design on the input space using fidelity parameter values $\mathbf{t} = (0.01, 0.001)$.
The cost function is approximated using the form $C(t) = bt_1^{-a_1}t_2^{-a_2}$, where $a_1, a_2, b \in \mathbb{R}^+$ are estimated using a separate set of 50 simulations. The available fidelity parameter space available for the active learning method is restricted to the hypercube $[0.075, 0.2] \times [0.0075, 0.02]$, while the design size $n$ for the single-fidelity model are chosen to match a certain cost. 

The results shown in Figure \ref{fig:wave_results} again demonstrate  the superiority of our proposed method (\texttt{AL LBM}) over other methods. In particular, the advantage of the active learning method over one-shot designs is more pronounced than in the previous case studies. In this case, our active learning is able to strategically select low-cost fidelity points that maximize predictive improvement, whereas one-shot designs cannot, leading to comparatively lower performance.


\begin{figure}[h]
    \centering
    \includegraphics[width = \linewidth]{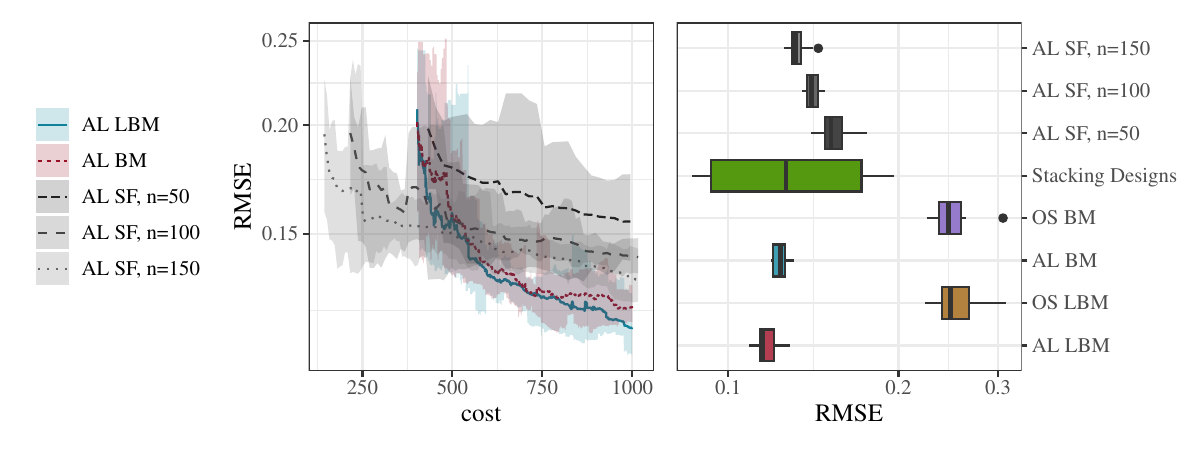}
    \caption{Left: RMSE (in logarithmic scale) for the wave equation case study with respect to the simulation cost. Solid lines indicate the average over 10 repetitions, while shaded regions represent the range. Right: Boxplots of the final RMSE in logarithmic scale for the wave equation case study across 10 repetitions.}
    \label{fig:wave_results}
\end{figure}

\section{Conclusion}\label{sec:conclusion}

Surrogate modeling for multi-fidelity experiments, particularly in the context of FEM simulations with continuous fidelity parameters, stands to benefit significantly from this new framework, which combines an adaptive kernel function that better captures the behavior of FEM outputs and an active learning method that can effectively leverage the continuous nature of the fidelity parameter. Through numerical studies and real case studies, we demonstrate that our active learning approach enhances the predictive performance of our model, while our proposed kernel offers greater flexibility and, in some cases, significantly outperforms other kernels.


We also studied the importance of the initial design in both the subsequent parameter estimation and the final predictive performance of the model. While the choice of initial design significantly impacts parameter estimation accuracy, its effect on the final prediction is less straightforward. The benefits of accurate parameter estimation must be balanced with the space-filling properties of the design. Further research is needed to explore this topic, particularly to account for varying sources of uncertainty, as discussed in \cite{haaland2018framework}, which provides upper bounds for prediction errors due to numerical approximation and parameter estimation.

Furthermore, a natural extension of our framework to multi-dimensional outputs, which are common in complex computer simulations, is to first reduce the dimensionality of the simulator outputs using techniques such as principal component decomposition \citep{higdon2008computer}, wavelet-based methods \citep{bayarri2007computer}, or modern deep learning approaches, such as autoencoders, to obtain a compact representation of high-dimensional outputs. Once projected onto a lower-dimensional space, our multi-fidelity surrogate model and active learning approach can be applied, with predictions subsequently mapped back to the original output space. 

Finally, it would be valuable to consider the transition costs between mesh configurations in FEM simulations, particularly the potentially high cost of moving from a coarse to a dense mesh, to better reflect real-world constraints. The look-ahead approach proposed in \cite{binois2019replication} could offer a promising direction for deciding whether to refine the mesh size or explore a new input location. We leave this as an avenue for future work.

\vspace{0.5cm}
\noindent\textbf{Supplemental Materials}
Additional supporting materials can be found in Supplemental Materials, including the proof of Theorem \ref{theorem:imspe_seq}, and the supporting figures for Sections \ref{sec:active_learning} and \ref{sec:numerical_study}. The \textsf{R} code and package for reproducing the results in Sections
\ref{sec:numerical_study} and \ref{sec:casestudy} are also provided.

\vspace{0.5cm}
\noindent\textbf{Data Availability Statement} The authors confirm that the data supporting the findings of this study are available within the article and its supplementary materials.

\bibliography{ref}

\end{document}


\maketitle

\spacingset{1}

\section{Details on Section 1}

\subsection{Asymptotic properties of the non-stationary kernel $K_{\gamma}$}

We first we demonstrate that that if any fidelity parameter $t_j$ of $\mathbf{t}$ is non-zero, then $\lim_{\mathbf{t}_{-j}\to \mathbf{0}} P(\delta(\mathbf{x}, \mathbf{t}) = 0) = 0$ given our multi-fidelity model, and the non-stationary kernel $K_{\gamma}$.

Observing that $\delta(\mathbf{x}, \mathbf{t}) \sim \mathcal{N}(\mu_{\delta}(\mathbf{x}, \mathbf{t}), \sigma^2K_{\gamma}(\mathbf{t}, \mathbf{t}))$, we only need to prove that if any fidelity parameter $t_j$ of $\mathbf{t}$ is non-zero, then $K_{\gamma}(\mathbf{t}, \mathbf{t}) > 0$. Since

\begin{align*}
    K_{\gamma}(\mathbf{t}, \mathbf{t}) & = \Big(\sum_{j = 1}^m (a_j t_j^{l_j})^{\frac{1}{\gamma}}\Big)^{\gamma} = \Vert \mathbf{a} \odot \mathbf{t}^{\mathbf{l}} \Vert_{\frac{1}{\gamma}}
\end{align*}

where $\mathbf{a} \odot \mathbf{t}^{\mathbf{l}} = (a_j t_j^{l_j})_{j=1}^m$, with $a_j > 0$ for all $j$, and $\displaystyle \Vert v \Vert_p = \big(\sum_{i=1}^n v_i^p\big)^{\frac{1}{p}}$ is the $p$-norm for vectors.

The monotonous property of the $p$-norm implies that $\Vert x \Vert_p \leq \Vert x \Vert_q$ for $p \leq q$, therefore we obtain that

\begin{align*}
    K_{\gamma}(\mathbf{t}, \mathbf{t}) & = \Vert \mathbf{a} \odot \mathbf{t}^{\mathbf{l}} \Vert_{\frac{1}{\gamma}} \leq \Vert \mathbf{a} \odot \mathbf{t}^{\mathbf{l}} \Vert_{\infty} = \max_{1 \leq j \leq m} a_j t_j^{l_j} > 0
\end{align*}

since $a_j > 0$ for all $j$ and there is a fidelity parameter $t_j > 0$ by assumption.

\subsection{Gradients of the log-likelihood}

The first step in deriving the gradient of the log-likelihood is to derive the gradient of $\mathbf{K}_{0,n}$ with respect of all its hyper-parameters. We remind the reader of the log-likelihood,
\begin{equation}\label{eq:MLE}
    l(\boldsymbol{\theta}; \mathbf{Y}) \propto - \frac{n-p}{2} \log \Big\{ \mathbf{Z}^T\left(\mathbf{K}_{0,n}^{-1} - \mathbf{K}_{0,n}^{-1}\mathbf{F_n}\mathbf{P}_n^{-1} \mathbf{F_n}^T \mathbf{K}_{0,n}^{-1} \right) \mathbf{Z} \Big\} 
    -\frac{1}{2} \log \vert \mathbf{K}_{0,n}\vert  -\frac{1}{2} \log \vert \mathbf{P}_n\vert
\end{equation}
with  $\mathbf{P}_n = \mathbf{F}_n^T \mathbf{K}_{0,n}^{-1}\mathbf{F}_n$.\par

We provide the gradients of the kernels $R_{\phi}$ associated with the input parameter $\mathbf{x}$ are readily available for the Gaussian kernel and Mat\'ern kernel with $\nu = 0.5, 1.5, 2.5$ below.\par

For the Gaussian kernel, defined as $R_{\boldsymbol{\phi}}(\mathbf{h}) = \mathrm{e}^{-\sum_{j=1}^d\phi_j^2h_j^2}$, the gradient with respect to $\phi_i$ is

$$\frac{\partial R_{\boldsymbol{\phi}}}{\partial \phi_i} = 2\phi_i h_i^2 R_{\boldsymbol{\phi}}(\mathbf{h})$$

For the Mat\'ern kerne kernel with $\nu = 0.5$, defined as $R_{\boldsymbol{\phi}}(\mathbf{h}) = \mathrm{e}^{-\sum_{j=1}^d\phi_j \vert h_j\vert}$, the gradient with respect to $\phi_i$ is

$$\frac{\partial R_{\boldsymbol{\phi}}}{\partial \phi_i} = \vert h_i\vert R_{\boldsymbol{\phi}}(\mathbf{h})$$

For the Mat\'ern kerne kernel with $\nu = 1.5$, defined as $R_{\boldsymbol{\phi}}(\mathbf{h})=\prod^d_{j=1}\Big( 1+\sqrt{3}\phi_j|h_j| \Big) \mathrm{e}^{-\sqrt{3}\phi_j|h_j|}$, the gradient with respect to $\phi_i$ is

$$\frac{\partial R_{\boldsymbol{\phi}}}{\partial \phi_i} = -3\phi_i h_i^2\mathrm{e}^{-\sqrt{3}\phi_i \vert h_i \vert} \Bigg(\prod_{j\neq i}\Big( 1+\sqrt{3}\phi_j|h_j| \Big) \mathrm{e}^{-\sqrt{3}\phi_j|h_j|}\Bigg)$$

For the Mat\'ern kerne kernel with $\nu = 2.5$, defined as $\displaystyle \prod^d_{j=1}\Big( 1+\sqrt{5}\phi_j|h_j|+\frac{5}{3} \phi_j^2 h_j^2 \Big) \mathrm{e}^{-\sqrt{5}\phi_j|h_j|}$, the gradient with respect to $\phi_i$ is

$$\frac{\partial R_{\boldsymbol{\phi}}}{\partial \phi_i} = -\frac53\big(\phi_i h_i^2 + \sqrt{5} \phi_i^2 \vert h_i \vert^3\big)\mathrm{e}^{-\sqrt{5}\phi_i \vert h_i \vert} \Bigg(\prod_{j\neq i}\Big( 1+\sqrt{5}\phi_j|h_j| + \frac{5}{3} \phi_j^2 h_j^2\Big) \mathrm{e}^{-\sqrt{5}\phi_j|h_j|}\Bigg)$$

We can then obtain the gradient of $K_t$ with respect to $\gamma$. To reduce the expression, we define $u_{1, j} = a_j t_{1,j}^{l_j}$ and $u_{2, j} = a_j t_{2,j}^{l_j}$.

\begin{align*}
\frac{\partial K_t(t_1, t_2)}{\partial \gamma} &= \frac12 \Bigg[\Big(\sum_{j=1}^m u_{1,j}^{\frac{1}{\gamma}} \Big)^{\gamma}\Bigg(\ln\Big(\sum_{j=1}^m u_{1,j}^{\frac{1}{\gamma}} \Big) - \frac{\sum_{j=1}^m \ln(u_{1,j})u_{1,j}^{\frac{1}{\gamma}}}{\gamma \sum_{j=1}^m u_{1,j}^{\frac{1}{\gamma}}} \Bigg) \\
& \qquad \qquad \qquad \qquad \qquad + \Big(\sum_{j=1}^m u_{2,j}^{\frac{1}{\gamma}} \Big)^{\gamma}\Bigg(\ln\Big(\sum_{j=1}^m u_{2,j}^{\frac{1}{\gamma}} \Big) - \frac{\sum_{j=1}^m \ln(u_{2,j})u_{2,j}^{\frac{1}{\gamma}}}{\gamma \sum_{j=1}^m u_{2,j}^{\frac{1}{\gamma}}} \Bigg) \\
& - \Big(\sum_{j=1}^m (u_{1,j}^{\frac{1}{2\gamma}} - u_{2,j}^{\frac{1}{2\gamma}})^2 \Big)^{\gamma}\Bigg(\ln\Big(\sum_{j=1}^m (u_{1,j}^{\frac{1}{2\gamma}} - u_{2,j}^{\frac{1}{2\gamma}})^2 \Big) \\
& \qquad \qquad \qquad \qquad \qquad \qquad - \frac{\sum_{j=1}^m (u_{1,j}^{\frac{1}{2\gamma}} - u_{2,j}^{\frac{1}{2\gamma}}\big(\ln(u_{1,j})u_{1,j}^{\frac{1}{\gamma}}-\ln(u_{2,j})u_{2,j}^{\frac{1}{\gamma}}\big)}{\gamma \sum_{j=1}^m (u_{1,j}^{\frac{1}{2\gamma}} - u_{2,j}^{\frac{1}{2\gamma}})^2} \Bigg)\Bigg] 
\end{align*}

It is then straightforward to obtain the gradient of $\mathbf{K}_{0,n}$ with respect to all the hyper-parameters.\par 
From here, we can obtain the gradients for $\mathbf{P}_n$ for any hyper-parameter $\theta$ by carefully performing matrix calculus \citep{petersen2008cookbook},
\begin{align*}
    \frac{\partial \mathbf{K}_{0,n}^{-1}}{\partial \theta} & = -\mathbf{K}_{0,n}^{-1}\frac{\partial\mathbf{K}_{0,n}}{\partial \theta}\mathbf{K}_{0,n}^{-1}\\
    \frac{\partial \mathbf{P}_n}{\partial \theta} & = \mathbf{F}_n^T \frac{\partial \mathbf{K}_{0,n}^{-1}}{\partial \theta}\mathbf{F}_n
\end{align*}

We can additionally find the gradients for the log of the determinants
\begin{align*}
    \frac{\partial \log \vert \mathbf{K}_{0,n}\vert}{\partial \theta} & = \mathrm{tr} \left(\mathbf{K}_{0,n}^{-1} \frac{\partial \mathbf{K}_{0,n}}{\partial \theta}\right) \\
    \frac{\partial \log \vert \mathbf{P}_n\vert}{\partial \theta} & = \mathrm{tr} \left(\mathbf{P}_n^{-1} \frac{\partial \mathbf{P}_n}{\partial \theta}\right) 
\end{align*}

Finally, for any hyper-parameter, we obtain the following gradient of the log likelihood

\begin{align}
    \frac{\partial l(\boldsymbol{\theta}; \mathbf{Y})}{\partial \theta} & = -\frac12 \mathrm{tr} \left(\mathbf{K}_{0,n}^{-1} \frac{\partial \mathbf{K}_{0,n}}{\partial \theta}\right) - \frac12 \mathrm{tr} \left(\mathbf{P}_n^{-1} \frac{\partial \mathbf{P}_n}{\partial \theta}\right) \nonumber \\
    & \quad  - \frac{n-p}{2}\cdot \frac{\mathbf{Z}^T\left(2\cdot\frac{\partial \mathbf{K}_{0,n}^{-1}}{\partial \theta}\mathbf{F_n}\mathbf{P}_n^{-1} \mathbf{F_n}^T \mathbf{K}_{0,n}^{-1} + \mathbf{K}_{0,n}^{-1}\mathbf{F_n}\frac{\partial \mathbf{P}_n}{\partial \theta} \mathbf{F_n}^T \mathbf{K}_{0,n}^{-1}\right)\mathbf{Z} - \mathbf{Z}^T\frac{\partial \mathbf{K}_{0,n}^{-1}}{\partial \theta}\mathbf{Z}}{\mathbf{Z}^T\left(\mathbf{K}_{0,n}^{-1} - \mathbf{K}_{0,n}^{-1}\mathbf{F_n}\mathbf{P}_n^{-1} \mathbf{F_n}^T \mathbf{K}_{0,n}^{-1} \right) \mathbf{Z}} \label{gradloglik}
\end{align}

\section{Proof of Theorem 1}

Following from \cite{binois2019replication}, we can obtain that that the IMSPE can take a closed-form expression, avoiding the computational burden of evaluating the integral over the whole domain.

\begin{lemma}\label{lemma:imspe_closed}
    Let $\mathbf{W}_n = (w(\mathbf{x}_i,\mathbf{x}_j))_{i,j}$ be a $N \times N$ matrix with \\ $w(\mathbf{x}_i,\mathbf{x}_j) =\int_{\mathbf{x}\in\mathcal{X}} K((\mathbf{x}_i,0),(\mathbf{x},0))K((\mathbf{x}_j,0),(\mathbf{x},0))d\mathbf{x}$.\\ Let $E =\int_{\mathbf{x}\in\mathcal{X}} K((\mathbf{x},0),(\mathbf{x},0))d\mathbf{x}$, let $ \mathbf{G}=(g_{i,j})_{i,j}$ be a $p \times p$ matrix such that $g_{i,j} = \int_{\mathbf{x}\in\mathcal{X}} f_{i}(\mathbf{x},0)f_j(\mathbf{x},0) d\mathbf{x}$, and let $ \mathbf{H}_n = (h_j(\mathbf{x}_i))_{i,j}$ be a $n \times p$ matrix such that $h_j(\mathbf{x}_i) = \int_{\mathbf{x}\in\mathcal{X}} k_n((\mathbf{x}_i,0),(\mathbf{x},0))f_j(\mathbf{x},0) d\mathbf{x}$. Then, 
    \begin{equation}\label{eq:imspe_closed}
            I_n = E - \mathrm{tr}(\mathbf{K}_n^{-1}\mathbf{W}_n) + \mathrm{tr}\left(\mathbf{M}_n\mathbf{W}_n\right) + \mathrm{tr}\left((\mathbf{F}_n^T\mathbf{K}_n^{-1}\mathbf{F}_n)^{-1}\mathbf{G}\right) - 2.\mathrm{tr}(\mathbf{P}_n\mathbf{H}_n)
    \end{equation}
    where $\mathbf{M}_n = \mathbf{K}_n^{-1}\mathbf{F}_n(\mathbf{F}_n^T\mathbf{K}_n^{-1}\mathbf{F}_n)^{-1}\mathbf{F}_n^T\mathbf{K}_n^{-1}$, and $\mathbf{P}_n = (\mathbf{F}_n^T\mathbf{K}_n^{-1}\mathbf{F}_n)^{-1}\mathbf{F}_n^T\mathbf{K}_n^{-1}$.
\end{lemma}
\begin{proof}
    Observing that 
    \begin{align*}
        I_n & = \int_{\mathbf{x}\in \mathcal{X}} \sigma_n^2(\mathbf{x},0) d\mathbf{x} \nonumber \\ 
        & = \int_{\mathbf{x}\in \mathcal{X}}  K\left((\mathbf{x},0),(\mathbf{x},0)\right) - \mathbf{k}_n(\mathbf{x},0)^T \mathbf{K}_n^{-1} \mathbf{k}_n(\mathbf{x},0) + \left(\boldsymbol{\gamma}_n(\mathbf{x},0)\right)^T(\mathbf{F}_n^T\mathbf{K}_n^{-1}\mathbf{F}_n)^{-1}\boldsymbol{\gamma}_n(\mathbf{x},0) d\mathbf{x} \nonumber \\
        & = \mathbb{E}\left[K\left((X,0),(X,0)\right)\right] - \mathbb{E}\left[\mathbf{k}_n(X,0)^T \mathbf{K}_n^{-1} \mathbf{k}_n(X,0)\right] + \mathbb{E}\left[\left(\boldsymbol{\gamma}_n(X,0)\right)^T(\mathbf{F}_n^T\mathbf{K}_n^{-1}\mathbf{F}_n)^{-1}\boldsymbol{\gamma}_n(X,0)\right]
    \end{align*}
    where $X$ is uniformly sampled on $\mathcal{X}$.\par
    Since $K\left((X,0),(X,0)\right) =  \sigma_1^2  R_{\boldsymbol{\phi}_1}(X,X) = \sigma_1^2$, we can set $E = \mathbb{E}\left[K\left((X,0),(X,0)\right)\right]$ aside since it will stay constant.\par
    \cite{binois2019replication} cleverly observe that  $$\mathbb{E}\left[\mathbf{k}_n(X,0)^T \mathbf{K}_n^{-1} \mathbf{k}_n(X,0)\right] = \mathrm{tr}\left(\mathbb{E}\left[\mathbf{K}_n^{-1} \mathbf{k}_n(X,0)\mathbf{k}_n(X,0)^T\right]\right) = \mathbb{E}\left[\mathbf{K}_n^{-1}(\mathbf{M} + \mathbf{m}\mathbf{m}^T)\right]$$ where $\mathbf{M} = \mathrm{Cov}(\mathbf{k}_n(X,0), \mathbf{k}_n(X,0))$, and $\mathbf{m} = \mathbb{E}\left[\mathbf{k}_n(X,0)\right]$. Observing that $\mathbf{W}_n = \mathbf{M} + \mathbf{m}\mathbf{m}^T$ yields the first part of the proof.\par
    We can extend this to the term $\mathbb{E}\left[\left(\boldsymbol{\gamma}_n(X,0)\right)^T(\mathbf{F}_n^T\mathbf{K}_n^{-1}\mathbf{F}_n)^{-1}\boldsymbol{\gamma}_n(X,0)\right]$ by first using that $\boldsymbol{\gamma}_n(X,0) = \left(\mathbf{f}(X,0) - \mathbf{F}_n^T\mathbf{K}_n^{-1}\mathbf{k}_n(X,0)\right)$ to obtain
    \begin{align*}
        \mathbb{E}\left[\left(\boldsymbol{\gamma}_n(X,0)\right)^T(\mathbf{F}_n^T\mathbf{K}_n^{-1}\mathbf{F}_n)^{-1}\boldsymbol{\gamma}_n(X,0)\right] & = \mathbb{E}\left[\mathbf{f}(X,0)^T(\mathbf{F}_n^T\mathbf{K}_n^{-1}\mathbf{F}_n)^{-1}\mathbf{f}(X,0)\right] \\
        & + \mathbb{E}\left[\mathbf{k}_n(X,0)^T\mathbf{K}_n^{-1}\mathbf{F}_n(\mathbf{F}_n^T\mathbf{K}_n^{-1}\mathbf{F}_n)^{-1}\mathbf{F}_n^T\mathbf{K}_n^{-1}\mathbf{k}_n(X,0)\right] - \\ 
        &  2\cdot \mathbb{E}\left[\mathbf{f}(X,0)^T(\mathbf{F}_n^T\mathbf{K}_n^{-1}\mathbf{F}_n)^{-1}\mathbf{F}_n^T\mathbf{K}_n^{-1}\mathbf{k}_n(X,0)\right]
    \end{align*}
    We can observe that
    \begin{align*}
        \mathbb{E}\left[\mathbf{f}(X,0)^T(\mathbf{F}_n^T\mathbf{K}_n^{-1}\mathbf{F}_n)^{-1}\mathbf{f}(X,0)\right] & = \mathrm{tr}\left((\mathbf{F}_n^T\mathbf{K}_n^{-1}\mathbf{F}_n)^{-1}\mathbb{E}\left[\mathbf{f}(X,0)\mathbf{f}(X,0)^T\right]\right)\\
        & =  \mathrm{tr}\left((\mathbf{F}_n^T\mathbf{K}_n^{-1}\mathbf{F}_n)^{-1}\mathbf{G}\right)
    \end{align*}
    \begin{align*}
        & \mathbb{E}\left[\mathbf{k}_n(X,0)^T\mathbf{K}_n^{-1}\mathbf{F}_n(\mathbf{F}_n^T\mathbf{K}_n^{-1} \mathbf{F}_n)^{-1}\mathbf{F}_n^T\mathbf{K}_n^{-1}\mathbf{k}_n(X,0)\right] \\
        & \qquad \qquad \qquad \qquad \qquad \qquad \qquad = \mathrm{tr}\left(\mathbf{K}_n^{-1}\mathbf{F}_n(\mathbf{F}_n^T\mathbf{K}_n^{-1}\mathbf{F}_n)^{-1}\mathbf{F}_n^T\mathbf{K}_n^{-1}\mathbb{E}\left[\mathbf{k}_n(X,0)\mathbf{k}_n(X,0)^T\right]\right) \\
        & \qquad \qquad \qquad \qquad \qquad \qquad \qquad =  \mathrm{tr}\left(\mathbf{M}_n\mathbf{W}_n\right)
    \end{align*}
    \begin{align*}
    \mathbb{E}\left[\mathbf{f}(X,0)^T(\mathbf{F}_n^T\mathbf{K}_n^{-1}\mathbf{F}_n)^{-1}\mathbf{F}_n^T\mathbf{K}_n^{-1}\mathbf{k}_n(X,0)\right] & = \mathrm{tr}\left((\mathbf{F}_n^T\mathbf{K}_n^{-1}\mathbf{F}_n)^{-1}\mathbf{F}_n^T\mathbf{K}_n^{-1}\mathbb{E}\left[\mathbf{k}_n(X,0)\mathbf{f}(X,0)^T\right]\right) \\
    & =  \mathrm{tr}\left(\mathbf{P}_n\mathbf{H}_n\right)
    \end{align*}
    
    The final form (A.1) follows from the three equations above.
\end{proof}

\vspace{10pt}

Fixing the first $n$ design points, we obtain the following forms for $\mathbf{K}_{n+1}$ and $\mathbf{W}_{n+1}$ using the new design point candidate $(\tilde{\mathbf{x}},\tilde{t})$.
$$\mathbf{K}_{n+1} = \begin{bmatrix}
\mathbf{K}_n & \mathbf{k}_n(\tilde{\mathbf{x}},\tilde{t}) \\
\mathbf{k}_n(\tilde{\mathbf{x}},\tilde{t})^T & K((\tilde{\mathbf{x}},\tilde{t}),(\tilde{\mathbf{x}},\tilde{t}))
\end{bmatrix},~~~~~~~~~~~~~ 
\mathbf{W}_{n+1} = \begin{bmatrix}
    \mathbf{W}_n & \mathbf{w}(\tilde{\mathbf{x}}) \\
    \mathbf{w}(\tilde{\mathbf{x}})^T & w(\tilde{\mathbf{x}},\tilde{\mathbf{x}})
\end{bmatrix}$$

where $\mathbf{w}(\tilde{\mathbf{x}}) = (w(\tilde{\mathbf{x}},\mathbf{x}_i))_{1\leq i \leq N}$.

The partitioned matrix inverse \citep{barnett1979matrix} gives,

$$\mathbf{K}_{n+1}^{-1} = \begin{bmatrix}
\mathbf{K}_n^{-1} +  \sigma_n^2(\tilde{\mathbf{x}},\tilde{t})\mathbf{a}(\tilde{\mathbf{x}},\tilde{t})\mathbf{a}(\tilde{\mathbf{x}},\tilde{t})^T &  \mathbf{a}(\tilde{\mathbf{x}},\tilde{t}) \\
\mathbf{a}(\tilde{\mathbf{x}},\tilde{t})^T & \sigma_n^2(\tilde{\mathbf{x}},\tilde{t})^{-1}
\end{bmatrix}$$
where $\mathbf{a}(\tilde{\mathbf{x}},\tilde{t}) = -\sigma_n^2(\tilde{\mathbf{x}},\tilde{t})^{-1}\mathbf{K}_n^{-1}\mathbf{k}_n(\tilde{\mathbf{x}},\tilde{t})$
\vspace{6pt}

Following the framework from \cite{binois2019replication}, we use the decomposition $\mathrm{tr}(\mathbf{K}_n^{-1}\mathbf{W}_n) = \mathbf{1}^T(\mathbf{K}_n^{-1} \odot \mathbf{W}_n)\mathbf{1}$, where $\odot$ represents the Hadamard (element-wise) product for matrices.

Combining these two results gives us
\begin{equation}
    \begin{split}
    \mathrm{tr}(\mathbf{K}_{n+1}^{-1}\mathbf{W}_{n+1}) & = \mathbf{1}^T(\mathbf{K}_{n+1}^{-1} \odot \mathbf{W}_{n+1})\mathbf{1} \\
    & = \mathbf{1}^T(\mathbf{K}_n^{-1} \odot \mathbf{W}_n)\mathbf{1} + \sigma_n^2(\tilde{\mathbf{x}},\tilde{t})\mathbf{a}(\tilde{\mathbf{x}},\tilde{t})^T\mathbf{W}_n \mathbf{a}(\tilde{\mathbf{x}},\tilde{t}) + 2\mathbf{w}(\tilde{\mathbf{x}})^T\mathbf{a}(\tilde{\mathbf{x}},\tilde{t}) + \sigma_n^2(\tilde{\mathbf{x}},\tilde{t})^{-1}w(\mathbf{x},\mathbf{x}) \\
    & = \mathrm{tr}(\mathbf{K}_n^{-1}\mathbf{W}_n) + \sigma_n^2(\tilde{\mathbf{x}},\tilde{t})\mathbf{a}(\tilde{\mathbf{x}},\tilde{t})^T\mathbf{W}_n \mathbf{a}(\tilde{\mathbf{x}},\tilde{t}) + 2\mathbf{w}(\tilde{\mathbf{x}})^T\mathbf{a}(\tilde{\mathbf{x}},\tilde{t}) + \sigma_n^2(\tilde{\mathbf{x}},\tilde{t})^{-1}w(\tilde{\mathbf{x}},\tilde{\mathbf{x}})  \label{eq:r1}
    \end{split}
\end{equation}

We can also obtain the following partitioned forms for $\mathbf{F}_{n+1}$ and $\mathbf{H}_{n+1}$,

$$\mathbf{F}_{n+1} = \begin{bmatrix}
\mathbf{F}_n \\
\mathbf{f}(\tilde{\mathbf{x}},\tilde{t})^T
\end{bmatrix},~~~~~~~~~~~~~ 
\mathbf{H}_{n+1} = \begin{bmatrix}
    \mathbf{H}_n \\
    \mathbf{h}(\tilde{\mathbf{x}})^T
\end{bmatrix}$$
where  $\mathbf{h}(\tilde{\mathbf{x}}) = \left(h_j(\tilde{\mathbf{x}})\right)_{1\leq j \leq p}$.\\

We use partitioned matrix equations to obtain that
\begin{equation}
\label{imspe_uk_a}
        \mathbf{F}_{n+1}^T\mathbf{K}_{n+1}^{-1}\mathbf{F}_{n+1} = \mathbf{F}_n^T\mathbf{K}_n^{-1}\mathbf{F}_n + \mathbf{v}\mathbf{v}^T
\end{equation}
where  $\displaystyle \mathbf{v} = \tilde{\sigma}_n(\tilde{\mathbf{x}}, \tilde{t})\mathbf{F}_n^T\mathbf{a}(\tilde{\mathbf{x}},\tilde{t}) + \left(\tilde{\sigma}_n(\tilde{\mathbf{x}}, \tilde{t})\right)^{-1}f(\tilde{\mathbf{x}}, \tilde{t})$, and $\tilde{\sigma}_n^2(\tilde{\mathbf{x}}, \tilde{t}) = K((\tilde{\mathbf{x}},\tilde{t}),(\tilde{\mathbf{x}},\tilde{t})) - \mathbf{k}_n(\tilde{\mathbf{x}},\tilde{t})^T \mathbf{K}_n^{-1} \mathbf{k}_n(\tilde{\mathbf{x}},\tilde{t})$.\\

The Sherman-Morrison formula \citep{sherman1950adjustment} on (\ref{imspe_uk_a}) leads to

\begin{equation}
\label{imspe_uk_b}
    (\mathbf{F}_{n+1}^T\mathbf{K}_{n+1}^{-1}\mathbf{F}_{n+1})^{-1} = (\mathbf{F}_n^T\mathbf{K}_n^{-1}\mathbf{F}_n)^{-1} - \mathbf{B}_n
\end{equation}
where $\displaystyle \mathbf{B}_n = \frac{(\mathbf{F}_{n}^T\mathbf{K}_{n}^{-1}\mathbf{F}_{n})^{-1} \mathbf{v}\mathbf{v}^T (\mathbf{F}_{n}^T\mathbf{K}_{n}^{-1}\mathbf{F}_{n})^{-1}}{1+\mathbf{v}^T (\mathbf{F}_{n}^T\mathbf{K}_{n}^{-1}\mathbf{F}_{n})^{-1}\mathbf{v}}$.\\

Additionally, we can observe that 
\begin{equation}
\label{imspe_uk_c}
    \mathbf{K}_{n+1}^{-1}\mathbf{F}_{n+1} = \begin{bmatrix}
    \mathbf{K}_n^{-1}\mathbf{F}_n  \\
    \mathbf{0}_p^T
\end{bmatrix} + \mathbf{U}_n
\end{equation}
where $\mathbf{U}_n = \begin{bmatrix}
    \mathbf{S}_n  \\
    \mathbf{T}_n
\end{bmatrix} = \begin{bmatrix}
    \tilde{\sigma}_n^2(\tilde{\mathbf{x}},\tilde{t})\mathbf{a}(\tilde{\mathbf{x}},\tilde{t})\mathbf{a}(\tilde{\mathbf{x}},\tilde{t})^T\mathbf{F}_n  + \mathbf{a}(\tilde{\mathbf{x}},\tilde{t})f(\tilde{\mathbf{x}},\tilde{t})^T  \\
    \mathbf{a}(\tilde{\mathbf{x}},\tilde{t})^T\mathbf{F}_n + \left(\tilde{\sigma}_n^2(\tilde{\mathbf{x}},\tilde{t})\right)^{-1}f(\tilde{\mathbf{x}},\tilde{t})^T
\end{bmatrix}$.\\

Using equations (\ref{imspe_uk_b}) and (\ref{imspe_uk_c}), we can obtain the following results:

\begin{flalign}
\label{imspe_uk_d}
        \mathbf{M}_{n+1} & = \begin{bmatrix}
        \mathbf{M}_n - \mathbf{K}_n^{-1}\mathbf{F}_n\mathbf{B}_n\mathbf{F}_n^T\mathbf{K}_n^{-1} & \mathbf{0}_n \\
        \mathbf{0}_n^T & 0
        \end{bmatrix} + \begin{bmatrix}
        \mathbf{K}_n^{-1}\mathbf{F}_n\left((\mathbf{F}_n^T\mathbf{K}_n^{-1}\mathbf{F}_n)^{-1} - \mathbf{B}_n\right)\mathbf{U}_n^T  \\
        \mathbf{0}_{n+1}^T
        \end{bmatrix} && \\\nonumber
        & ~ + \begin{bmatrix}
        \mathbf{K}_n^{-1}\mathbf{F}_n\left((\mathbf{F}_n^T\mathbf{K}_n^{-1}\mathbf{F}_n)^{-1} - \mathbf{B}_n\right)\mathbf{U}_n^T  \\
        \mathbf{0}_{n+1}^T
        \end{bmatrix}^T + \mathbf{U}_n\left((\mathbf{F}_n^T\mathbf{K}_n^{-1}\mathbf{F}_n)^{-1} - \mathbf{B}_n\right)\mathbf{U}_n^T &&
\end{flalign}

\begin{flalign}
\label{imspe_uk_e}
    \mathbf{P}_{n+1} = \begin{bmatrix}
        \mathbf{P}_n & \mathbf{0}_{p}
    \end{bmatrix} - \begin{bmatrix}
        \mathbf{B}_n\mathbf{F}_n^T\mathbf{K}_n^{-1} &
        \mathbf{0}_{p}
    \end{bmatrix} + \left((\mathbf{F}_n^T\mathbf{K}_n^{-1}\mathbf{F}_n)^{-1} - \mathbf{B}_n\right)\mathbf{U}_n^T&&
\end{flalign}

From all these derivations, and using the results from Lemma 1, it follows that

\begin{equation}
    \begin{split}
        I_{n+1}(\tilde{\mathbf{x}}, \tilde{t}) & = E - \mathrm{tr}(\mathbf{K}_{n+1}^{-1}\mathbf{W}_{n+1}) + \mathrm{tr}\left(\mathbf{M}_{n+1}\mathbf{W}_{n+1}\right) + \mathrm{tr}\left((\mathbf{F}_{n+1}^T\mathbf{K}_{n+1}^{-1}\mathbf{F}_{n+1})^{-1}\mathbf{G}\right) - 2.\mathrm{tr}(\mathbf{P}_{n+1}\mathbf{H}_{n+1}) \\
        & = E - \mathrm{tr}(\mathbf{K}_n^{-1}\mathbf{W}_n) + \mathrm{tr}(\mathbf{M}_n\mathbf{W}_n)+ \mathrm{tr}\left((\mathbf{F}_n^T\mathbf{K}_n^{-1}\mathbf{F}_n)^{-1}\mathbf{G}\right) - 2.\mathrm{tr}(\mathbf{P}_n\mathbf{H}_n) - R_{n+1}(\tilde{\mathbf{x}}, \tilde{t}) \\
        & = I_n - R_{n+1}(\tilde{\mathbf{x}}, \tilde{t})
    \end{split}
\end{equation}

To fully develop $R_{n+1}(\tilde{\mathbf{x}}, \tilde{t})$, we write it as a sum of 4 elements : $$R_{n+1}(\tilde{\mathbf{x}}, \tilde{t}) = R^{(1)}_{n+1}(\tilde{\mathbf{x}}, \tilde{t}) + R^{(2)}_{n+1}(\tilde{\mathbf{x}}, \tilde{t}) + R^{(3)}_{n+1}(\tilde{\mathbf{x}}, \tilde{t}) + R_{n+1}^{(4)}(\tilde{\mathbf{x}}, \tilde{t})$$
such that:
\begin{flalign*}
    R^{(1)}_{n+1}(\tilde{\mathbf{x}}, \tilde{t}) & = \mathrm{tr}(\mathbf{K}_{n+1}^{-1}\mathbf{W}_{n+1}) - \mathrm{tr}(\mathbf{K}_n^{-1}\mathbf{W}_n) && \\
    & = \tilde{\sigma}_n^2(\tilde{\mathbf{x}},\tilde{t})\mathbf{a}(\tilde{\mathbf{x}},\tilde{t})^T\mathbf{W}_n \mathbf{a}(\tilde{\mathbf{x}},\tilde{t}) + 2\mathbf{w}(\tilde{\mathbf{x}})^T\mathbf{a}(\tilde{\mathbf{x}},\tilde{t}) + \tilde{\sigma}_n^2(\tilde{\mathbf{x}},\tilde{t})^{-1}w(\tilde{\mathbf{x}},\tilde{\mathbf{x}}) &&
\end{flalign*}

\begin{flalign*}
    R^{(2)}_{n+1}(\tilde{\mathbf{x}}, \tilde{t}) & = \mathrm{tr}(\mathbf{M}_{n+1}\mathbf{W}_{n+1}) - \mathrm{tr}(\mathbf{M}_n\mathbf{W}_n)  && \\
    & = \mathrm{tr}\left(\mathbf{K}_n^{-1}\mathbf{F}_n\mathbf{B}_n\mathbf{F}_n^T\mathbf{K}_n^{-1}\mathbf{W}_n\right) - \mathrm{tr}\left((\mathbf{S}_n + 2\mathbf{K}_n^{-1}\mathbf{F}_n)\left((\mathbf{F}_n^T\mathbf{K}_n^{-1}\mathbf{F}_n)^{-1} - \mathbf{B}_n\right)\mathbf{S}_n^T\mathbf{W}_n\right) && \\
    & ~~~~ - 2.\mathbf{T}_n\left((\mathbf{F}_n^T\mathbf{K}_n^{-1}\mathbf{F}_n)^{-1} - \mathbf{B}_n\right)(\mathbf{S}_n^T + \mathbf{F}_n^T\mathbf{K}_n^{-1})\mathbf{w}(\tilde{\mathbf{x}}) - \mathbf{T}_n\left((\mathbf{F}_n^T\mathbf{K}_n^{-1}\mathbf{F}_n)^{-1} - \mathbf{B}_n\right)\mathbf{T}_n^T w(\tilde{\mathbf{x}},\tilde{\mathbf{x}}) &&
 \end{flalign*}

\begin{flalign*}
    R^{(3)}_{n+1}(\tilde{\mathbf{x}}, \tilde{t}) & = \mathrm{tr}\left((\mathbf{F}_{n+1}^T\mathbf{K}_{n+1}^{-1}\mathbf{F}_{n+1})^{-1}\mathbf{G}\right) - \mathrm{tr}\left((\mathbf{F}_n^T\mathbf{K}_n^{-1}\mathbf{F}_n)^{-1}\mathbf{G}\right) && \\
    & = \mathrm{tr}(\mathbf{B}_n\mathbf{G})
\end{flalign*}

\begin{flalign*}
    R^{(4)}_{n+1}(\tilde{\mathbf{x}}, \tilde{t}) & = -2\Big[\mathrm{tr}(\mathbf{P}_{n+1}\mathbf{H}_{n+1}) - \mathrm{tr}(\mathbf{P}_n\mathbf{H}_n)\Big] && \\
    & = -2\Big[\mathrm{tr}\left(\mathbf{B}_n\mathbf{F}_n^T\mathbf{K}_n^{-1}\mathbf{H}_n\right) \\ 
    & \qquad \qquad - \mathrm{tr}\left(\left((\mathbf{F}_n^T\mathbf{K}_n^{-1}\mathbf{F}_n)^{-1} - \mathbf{B}_n\right)\mathbf{S}_n^T\mathbf{H}_n\right) - \mathbf{h}(\tilde{\mathbf{x}})^T\left((\mathbf{F}_n^T\mathbf{K}_n^{-1}\mathbf{F}_n)^{-1} - \mathbf{B}_n\right)\mathbf{T}_n^T\Big]
\end{flalign*}

More importantly, we can observe that given that the hyper-parameters remain the same, $\mathbf{K}_n^{-1}$ and $\mathbf{W}_n$ do not need to be reevaluated, leading to significant improvements in terms of computational efficiency. The cost complexity is then only $\mathcal{O}(n^2)$, considering that for all the above calculations, we can use the faster form $\mathrm{tr}(AB) = \mathbf{1}^T(A \odot B) \mathbf{1}$, for $A,B$ symmetric matrices.

\section{Closed form expressions in Theorem 1}

The closed-form expressions for $\mathbf{W}_n$ can be found in \cite{binois2019replication}.\par
\vspace{10pt}
We will detail the expressions of $g_{i,j}$ and $h_j(\mathbf{x}_i)$ in the cases  $f_j(\mathbf{x}, t) = 1,~f_j(\mathbf{x}, t) = x_{k_j},~f_j(\mathbf{x}, t) = x_{k_j}^2$. If using orthogonal regression functions such as Legendre polynomials, one can simply use linear combinations of those results. We will use $\mathcal{D} = [0,1]^d$ for simplicity.\par

\subsection{Derivation of $h_j(\mathbf{x}_i)$}

We only consider here the Gaussian kernel, and the Mat\'ern kernel with $\nu \in \{0.5, 1.5, 2.5\}$. All these kernels are separable, such that $k((\mathbf{x}_1, \mathbf{0}), (\mathbf{x}_2, \mathbf{0})) = \sigma^2\prod_{r = 1}^d k_r((x_{1,r}, 0),(x_{2,r},0))$. Since we always set the fidelity parameter at 0, we simplify notation and write $k(\mathbf{x}_1, \mathbf{x}_2) = \sigma^2\prod_{r = 1}^d k_r(x_{1,r},x_{2,r})$.\par
\vspace{5pt}
We first define the following quantities, where $1 \leq r \leq d$,
\begin{flalign*}
    I^{(1)}_{i,r} & := \int_0^1 k_r(x_{i,r},x) \mathrm{d}x  && \\
    I^{(2)}_{i,r} & = \int_0^1 x k_r(x_{i,r},x) \mathrm{d}x && \\
    I^{(3)}_{i,r} & = \int_0^1 x^2 k_r(x_{i,r},x) \mathrm{d}x
\end{flalign*}

These define the component-wise integrals that form $h_j(\mathbf{x}_i)$. Using these, we can simply develop the integrals for different regression functions $f_j(\mathbf{x}, \mathbf{t})$. Since $\mathbf{t}$ is once again set at 0, we simplify it as $f_j(\mathbf{x})$.

\begin{itemize}
    \item For $f_j(\mathbf{x}) = 1$, and with $\mathbf{x}_i = (\mathbf{x}_{i,1},\cdots,\mathbf{x}_{i,d})^T$,
    \begin{flalign*}
        h_j(\mathbf{x}_i) & = \int_{\mathbf{x}\in \mathcal{D}} k(\mathbf{x}_i,\mathbf{x})f_j(\mathbf{x})d\mathbf{x} = \prod_{r=1}^d I^{(1)}_{i,r} &&
    \end{flalign*}

    \item For $f_j(\mathbf{x}) = x_{r_j}$,
    \begin{flalign*}
        h_j(\mathbf{x}_i) & = \int_{\mathbf{x}\in \mathcal{D}} k(\mathbf{x}_i,\mathbf{x})f_j(\mathbf{x})d\mathbf{x} = \sigma_1^2\left(\int_0^1 k_{r_j}(x_{i,r_j}, x) x \mathrm{d}x\right)\left(\prod_{r\neq r_j} \int_0^1 k_r(x_{i,r}, x) \mathrm{d}x\right) && \\
        & = \sigma_1^2 I^{(2)}_{i,r_j}\left(\prod_{r\neq r_j} I^{(1)}_{i,r}\right)
    \end{flalign*}

    \item For $f_j(\mathbf{x}, t) = x_{r_j}^2$,
    \begin{flalign*}
        h_j(\mathbf{x}_i) & = \int_{\mathbf{x}\in \mathcal{D}} k(\mathbf{x}_i,\mathbf{x})f_j(\mathbf{x})d\mathbf{x} = \sigma_1^2\left(\int_0^1 k_{r_j}(x_{i,r_j}, x) x^2 \mathrm{d}x\right)\left(\prod_{r\neq r_j} \int_0^1 k_r(x_{i,r}, x) \mathrm{d}x\right) && \\
        & = \sigma_1^2 I^{(3)}_{i,r_j}\left(\prod_{r\neq r_j} I^{(1)}_{i,r}\right)
    \end{flalign*}
    
    \item For $f_j(\mathbf{x}, t) = x_{r_{j,1}} x_{r_{j,2}}$,
    \begin{flalign*}
        h_j(\mathbf{x}_i) & = \int_{\mathbf{x}\in \mathcal{D}} k(\mathbf{x}_i,\mathbf{x})f_j(\mathbf{x})d\mathbf{x} && \\
        & = \sigma_1^2\left(\int_0^1 k_{r_{j,1}}(x_{i,r_{j,1}}, x) x \mathrm{d}x\right) \left(\int_0^1 k_{r_{j,2}}(x_{i,r_{j,2}}, x) x \mathrm{d}x\right) \left(\prod_{r\notin \{r_{j,1},r_{j,2}\}} \int_0^1 k_r(x_{i,r}, x) \mathrm{d}x\right) && \\
        & = \sigma_1^2 I^{(2)}_{i,r_{j,1}} I^{(2)}_{i,r_{j,2}}\left(\prod_{r\notin \{r_{j,1},r_{j,2}\}} I^{(1)}_{i,r}\right)
    \end{flalign*}
\end{itemize}

We now need to obtain the closed form expressions for $I^{(1)}_{i,r},I^{(2)}_{i,r}, I^{(3)}_{i,r}$ for all the kernels considered in this work. The calculations mostly rely on integration by parts, as well as separating the domain in 2 parts ($x < x_{i,r}$ and $x_{i,r} < x$) for the Mat\'ern kernel.

\subsubsection{Gaussian kernel}

The closed-form expressions for the Gaussian kernel are 

\begin{flalign*}
    I^{(1)}_{i,r} & = \int_0^1 \mathrm{e}^{-\phi_r^2(x - x_{i,r})^2}\mathrm{d}x = \frac{\sqrt{\pi}}{\phi_r} \Big[\Phi(\sqrt{2}\phi_k(1-x_{i,r})) - \Phi(-\sqrt{2}\phi_r x_{i,r}) \Big] && \\
    I^{(2)}_{i,r} & = \int_0^1 x\mathrm{e}^{-\phi_{r}^2(x - x_{i,r})^2}\mathrm{d}x = \frac{1}{2\phi_r^2}\left(\mathrm{e}^{-\phi_{r}^2x_{i,r}^2} - \mathrm{e}^{-\phi_{r}^2(1-x_{i,r})^2}\right) + x_{i,r} I^{(1)}_{i,r} && \\
    I^{(3)}_{i,r} & = \int_0^1 x^2 \mathrm{e}^{-\phi_{r}^2(x - x_{i,r})^2}\mathrm{d}x && \\
    & = \frac{1}{2\phi_r^2}\big(-x_{i,r} \mathrm{e}^{-\phi_{r}^2 x_{i,r}^2} - (1 - x_{i,r}) \mathrm{e}^{-\phi_{r}^2 (1 - x_{i,r})^2}\big) + 2 x_{i,r}I^{(2)}_{i,r} + \Big(\frac{1}{2\phi_r^2} - x_{i,r}^2\Big)I^{(1)}_{i,r}
\end{flalign*}

\subsubsection{Mat\'ern kernel with $\nu = 0.5$}

The closed-form expressions for the Mat\'ern kernel with $\nu = 0.5$ are

\begin{flalign*}
    I^{(1)}_{i,r} & = \int^{x_{i,r}}_{0} \mathrm{e}^{-{\phi}_{r} \left|x_{i,r} - x\right|} \, \mathrm{d}x = \frac{2}{{\phi}_{r}} - \frac{\mathrm{e}^{-{\phi}_{r}}}{{\phi}_{r}} - \frac{\mathrm{e}^{\phi_{r}^{2}(1 - x_{i,r})}}{\phi_r} && \\
    I^{(2)}_{i,r} & = \int^{x_{i,r}}_{0} x\mathrm{e}^{-{\phi}_{r} \left|x_{i,r} - x\right|} \, \mathrm{d}x = \frac{1}{\phi_r}\Bigg(2x_{i,r} - \mathrm{e}^{-{\phi}_{r}(1 - x_{i,r})} + \frac{\mathrm{e}^{-{\phi}_{r} x_{i,r}}}{{\phi}_{r}} - \frac{ \mathrm{e}^{-{\phi}_{r}(1 - x_{i,r})}}{{\phi}_{r}} \Bigg)  && \\
    I^{(3)}_{i,r} & = \int^{x_{i,r}}_{0} x^2 \mathrm{e}^{-{\phi}_{r} \left|x_{i,r} - x\right|} \, \mathrm{d}x = \frac{1}{\phi_r} \Bigg(2 x^2_{i,r} - \mathrm{e}^{-\phi_r(1-x_{i,r})}  && \\
    & \qquad \qquad \qquad \qquad \qquad \qquad \qquad \qquad - \frac{2}{\phi_r} \Big(\mathrm{e}^{-\phi_r(1 - x_{i,r})}\Big(1 + \frac{1}{\phi_r}\Big)- \frac{2}{\phi_r} + \frac{\mathrm{e}^{-\phi_r x_{i,r}}}{\phi_r} \Big) \Bigg) 
\end{flalign*}

\subsubsection{Mat\'ern kernel with $\nu = 1.5$}

To reduce expression, we define

\begin{flalign*}
    A^{(1)}_{i,r} & := \int_0^{x_{i,r}} \mathrm{e}^{-\sqrt{3}\phi_r(x_{i,r}-x)} \mathrm{d}x = \frac{1}{\sqrt{3}\phi_r} - \frac{\mathrm{e}^{-\sqrt{3}\phi_rx_{i,r}}}{\sqrt{3}\phi_r} && \\
    A^{(2)}_{i,r} & := \int_0^{x_{i,r}} x \mathrm{e}^{-\sqrt{3}\phi_r(x_{i,r}-x)} \mathrm{d}x = \frac{x_{i,r}}{\sqrt{3}\phi_r} - \frac{A^{(1)}_{i,r}}{\sqrt{3}\phi_r} && \\
    A^{(3)}_{i,r} & := \int_0^{x_{i,r}} x^2 \mathrm{e}^{-\sqrt{3}\phi_r(x_{i,r}-x)} \mathrm{d}x = \frac{x_{i,r}^2}{\sqrt{3}\phi_r} - \frac{2A^{(2)}_{i,r}}{\sqrt{3}\phi_r}  && \\
    A^{(4)}_{i,r} & := \int_0^{x_{i,r}} x^3 \mathrm{e}^{-\sqrt{3}\phi_r(x_{i,r}-x)} \mathrm{d}x = \frac{x_{i,r}^3}{\sqrt{3}\phi_r} - \frac{3A^{(3)}_{i,r}}{\sqrt{3}\phi_r} 
\end{flalign*}

\begin{flalign*}
    B^{(1)}_{i,r} & := \int_{x_{i,r}}^1 \mathrm{e}^{-\sqrt{3}\phi_r(x-x_{i,r})} \mathrm{d}x = \frac{1}{\sqrt{3}\phi_r} - \frac{\mathrm{e}^{-\sqrt{3}\phi_r(1-x_{i,r})}}{\sqrt{3}\phi_r} && \\
    B^{(2)}_{i,r} & := \int_{x_{i,r}}^1 x \mathrm{e}^{-\sqrt{3}\phi_r(x-x_{i,r})} \mathrm{d}x = \frac{x_{i,r}}{\sqrt{3}\phi_r} - \frac{\mathrm{e}^{-\sqrt{3}\phi_r(1-x_{i,r})}}{\sqrt{3}\phi_r} + \frac{B^{(1)}_{i,r}}{\sqrt{3}\phi_r} && \\
    B^{(3)}_{i,r} & := \int_{x_{i,r}}^1 x^2 \mathrm{e}^{-\sqrt{3}\phi_r(x-x_{i,r})} \mathrm{d}x = \frac{x_{i,r}^2}{\sqrt{3}\phi_r} - \frac{\mathrm{e}^{-\sqrt{3}\phi_r(1-x_{i,r})}}{\sqrt{3}\phi_r} + \frac{2 B^{(2)}_{i,r}}{\sqrt{3}\phi_r} && \\
    B^{(4)}_{i,r} & := \int_{x_{i,r}}^1 x^3 \mathrm{e}^{-\sqrt{3}\phi_r(x-x_{i,r})} \mathrm{d}x = \frac{x_{i,r}^3}{\sqrt{3}\phi_r} - \frac{\mathrm{e}^{-\sqrt{3}\phi_r(1-x_{i,r})}}{\sqrt{3}\phi_r} + \frac{3 B^{(3)}_{i,r}}{\sqrt{3}\phi_r}
\end{flalign*}

We then obtain the following closed-form expressions

\begin{flalign*}
    I^{(1)}_{i,r} & = \int_0^1 (1 + \sqrt{3}\phi_r|x_{i,r}-x|)\mathrm{e}^{-\sqrt{3}{\phi}_{r} \left|x_{i,r} - x\right|} \mathrm{d}x && \\
    & = (1 + \sqrt{3}\phi_rx_{i,r})A^{(1)}_{i,r} - \sqrt{3}\phi_r A^{(2)}_{i,r} + (1 - \sqrt{3}\phi_rx_{i,r})B^{(1)}_{i,r} + \sqrt{3}\phi_r B^{(2)}_{i,r} && \\
    I^{(2)}_{i,r} & = \int_0^1 x \Big(1 + \sqrt{3}\phi_r|x_{i,r}-x|\Big)\mathrm{e}^{-\sqrt{3}{\phi}_{r} \left|x_{i,r} - x\right|} \mathrm{d}x && \\
    & = (1 + \sqrt{3}\phi_rx_{i,r})A^{(2)}_{i,r} - \sqrt{3}\phi_r A^{(3)}_{i,r} + (1 - \sqrt{3}\phi_rx_{i,r})B^{(2)}_{i,r} + \sqrt{3}\phi_r B^{(3)}_{i,r} && \\
    I^{(3)}_{i,r} & = \int_0^1 x^2 (1 + \sqrt{3}\phi_r|x_{i,r}-x|)\mathrm{e}^{-\sqrt{3}{\phi}_{r} \left|x_{i,r} - x\right|} \mathrm{d}x && \\
    & = (1 + \sqrt{3}\phi_rx_{i,r})A^{(3)}_{i,r} - \sqrt{3}\phi_r A^{(4)}_{i,r} + (1 - \sqrt{3}\phi_rx_{i,r})B^{(3)}_{i,r} + \sqrt{3}\phi_r B^{4}_{i,r}
\end{flalign*}

\subsubsection{Mat\'ern kernel with $\nu = 2.5$}

To reduce expression, we define

\begin{flalign*}
    A^{(1)}_{i,r} & := \int_0^{x_{i,r}} \mathrm{e}^{-\sqrt{5}\phi_r(x_{i,r}-x)} \mathrm{d}x = \frac{1}{\sqrt{5}\phi_r} - \frac{\mathrm{e}^{-\sqrt{5}\phi_rx_{i,r}}}{\sqrt{5}\phi_r} && \\
    A^{(2)}_{i,r} & := \int_0^{x_{i,r}} x \mathrm{e}^{-\sqrt{5}\phi_r(x_{i,r}-x)} \mathrm{d}x = \frac{x_{i,r}}{\sqrt{5}\phi_r} - \frac{A^{(1)}_{i,r}}{\sqrt{5}\phi_r} && \\
    A^{(3)}_{i,r} & := \int_0^{x_{i,r}} x^2 \mathrm{e}^{-\sqrt{5}\phi_r(x_{i,r}-x)} \mathrm{d}x = \frac{x_{i,r}^2}{\sqrt{5}\phi_r} - \frac{2A^{(2)}_{i,r}}{\sqrt{5}\phi_r} && \\
    A^{(4)}_{i,r} & := \int_0^{x_{i,r}} x^3 \mathrm{e}^{-\sqrt{5}\phi_r(x_{i,r}-x)} \mathrm{d}x = \frac{x_{i,r}^3}{\sqrt{5}\phi_r} - \frac{3A^{(3)}_{i,r}}{\sqrt{5}\phi_r} && \\
    A^{(5)}_{i,r} & := \int_0^{x_{i,r}} x^4 \mathrm{e}^{-\sqrt{5}\phi_r(x_{i,r}-x)} \mathrm{d}x = \frac{x_{i,r}^4}{\sqrt{5}\phi_r} - \frac{4A^{(4)}_{i,r}}{\sqrt{5}\phi_r} 
\end{flalign*}

\begin{flalign*}
    B^{(1)}_{i,r} & := \int_{x_{i,r}}^1 \mathrm{e}^{-\sqrt{5}\phi_r(x-x_{i,r})} \mathrm{d}x = \frac{1}{\sqrt{5}\phi_r} - \frac{\mathrm{e}^{-\sqrt{5}\phi_r(1-x_{i,r})}}{\sqrt{5}\phi_r} && \\
    B^{(2)}_{i,r} & := \int_{x_{i,r}}^1 x \mathrm{e}^{-\sqrt{5}\phi_r(x-x_{i,r})} \mathrm{d}x = \frac{x_{i,r}}{\sqrt{5}\phi_r} - \frac{\mathrm{e}^{-\sqrt{5}\phi_r(1-x_{i,r})}}{\sqrt{5}\phi_r} + \frac{B^{(1)}_{i,r}}{\sqrt{5}\phi_r} && \\
    B^{(3)}_{i,r} & := \int_{x_{i,r}}^1 x^2 \mathrm{e}^{-\sqrt{5}\phi_r(x-x_{i,r})} \mathrm{d}x = \frac{x_{i,r}^2}{\sqrt{5}\phi_r} - \frac{\mathrm{e}^{-\sqrt{5}\phi_r(1-x_{i,r})}}{\sqrt{5}\phi_r} + \frac{2 B^{(2)}_{i,r}}{\sqrt{5}\phi_r} && \\
    B^{(4)}_{i,r} & := \int_{x_{i,r}}^1 x^3 \mathrm{e}^{-\sqrt{5}\phi_r(x-x_{i,r})} \mathrm{d}x = \frac{x_{i,r}^3}{\sqrt{5}\phi_r} - \frac{\mathrm{e}^{-\sqrt{5}\phi_r(1-x_{i,r})}}{\sqrt{5}\phi_r} + \frac{3 B^{(3)}_{i,r}}{\sqrt{5}\phi_r} && \\
    B^{(5)}_{i,r} & := \int_{x_{i,r}}^1 x^4 \mathrm{e}^{-\sqrt{5}\phi_r(x-x_{i,r})} \mathrm{d}x = \frac{x_{i,r}^4}{\sqrt{5}\phi_r} - \frac{\mathrm{e}^{-\sqrt{5}\phi_r(1-x_{i,r})}}{\sqrt{5}\phi_r} + \frac{4 B^{(4)}_{i,r}}{\sqrt{5}\phi_r}
\end{flalign*}

We then obtain the following closed-form expressions

\begin{flalign*}
    I^{(1)}_{i,r} & = \int_0^1 \Big(1 + \sqrt{5}\phi_r|x_{i,r}-x| + \frac53 \phi_r^2(x_{i,r}-x)^2\Big)\mathrm{e}^{-\sqrt{5}{\phi}_{r} \left|x_{i,r} - x\right|} \mathrm{d}x && \\
    & = (1 + \sqrt{5}\phi_rx_{i,r})A^{(1)}_{i,r} - (\sqrt{5}\phi_r + \frac{10}{3}\phi_r^2 x_{i,r}) A^{(2)}_{i,r} + \frac53 \phi_r^2 A^{(3)}_{i,r} + (1 - \sqrt{5}\phi_rx_{i,r})B^{(1)}_{i,r} && \\
    & \qquad  \qquad \qquad \qquad \qquad \qquad \qquad \qquad \qquad \qquad \qquad+ (\sqrt{5}\phi_r - \frac{10}{3}\phi_r^2 x_{i,r}) B^{(2)}_{i,r} + \frac53 \phi_r^2 B^{(3)}_{i,r} && \\
    I^{(2)}_{i,r} & = \int_0^1 x \Big(1 + \sqrt{5}\phi_r|x_{i,r}-x| + \frac53 \phi_r^2(x_{i,r}-x)^2\Big)\mathrm{e}^{-\sqrt{5}{\phi}_{r} \left|x_{i,r} - x\right|} \mathrm{d}x && \\
    & = (1 + \sqrt{5}\phi_rx_{i,r})A^{(2)}_{i,r} - (\sqrt{5}\phi_r + \frac{10}{3}\phi_r^2 x_{i,r}) A^{(3)}_{i,r} + \frac53 \phi_r^2 A^{(4)}_{i,r} + (1 - \sqrt{5}\phi_rx_{i,r})B^{(2)}_{i,r} && \\
    & \qquad  \qquad \qquad \qquad \qquad \qquad \qquad \qquad \qquad \qquad \qquad+ (\sqrt{5}\phi_r - \frac{10}{3}\phi_r^2 x_{i,r}) B^{(3)}_{i,r} + \frac53 \phi_r^2 B^{(4)}_{i,r} && \\
    I^{(3)}_{i,r} & = \int_0^1 x^2 \Big(1 + \sqrt{5}\phi_r|x_{i,r}-x| + \frac53 \phi_r^2(x_{i,r}-x)^2\Big)\mathrm{e}^{-\sqrt{5}{\phi}_{r} \left|x_{i,r} - x\right|} \mathrm{d}x && \\
    & = (1 + \sqrt{5}\phi_rx_{i,r})A^{(3)}_{i,r} - (\sqrt{5}\phi_r + \frac{10}{3}\phi_r^2 x_{i,r}) A^{(4)}_{i,r} + \frac53 \phi_r^2 A^{(5)}_{i,r} + (1 - \sqrt{5}\phi_rx_{i,r})B^{(3)}_{i,r} && \\
    & \qquad  \qquad \qquad \qquad \qquad \qquad \qquad \qquad \qquad \qquad \qquad+ (\sqrt{5}\phi_r - \frac{10}{3}\phi_r^2 x_{i,r}) B^{(4)}_{i,r} + \frac53 \phi_r^2 B^{(5)}_{i,r}
\end{flalign*}

\subsubsection{Derivation of $g_{i,j}$}

We then derive the closed form expressions for $g_{i,j}$:
\begin{itemize}
    \item For $f_i(\mathbf{x}, t) = 1$ and $f_j(\mathbf{x}, t) = 1$, them $g_{i,j} = 1$
    \item For $f_i(\mathbf{x}, t) = 1$ and $f_j(\mathbf{x}, t) = x_{k_j}$, then $g_{i,j} = \frac12$
    \item For $f_i(\mathbf{x}, t) = 1$ and $f_j(\mathbf{x}, t) = x_{k_j}^2$, then $g_{i,j} = \frac13$
    \item For $f_i(\mathbf{x}, t) = x_{k_i}$ and $f_j(\mathbf{x}, t) = x_{k_j}$ with $i\neq j$, then $g_{i,j} = \frac14$
    \item For $f_i(\mathbf{x}, t) = x_{k_i}$ and $f_j(\mathbf{x}, t) = x_{k_j}$ with $i = j$, then $g_{i,j} = \frac13$
    \item For $f_i(\mathbf{x}, t) = x_{k_i}$ and $f_j(\mathbf{x}, t) = x_{k_j}^2$ with $i \neq j$, then $g_{i,j} = \frac16$
    \item For $f_i(\mathbf{x}, t) = x_{k_i}$ and $f_j(\mathbf{x}, t) = x_{k_j}^2$ with $i = j$, then $g_{i,j} = \frac14$
    \item For $f_i(\mathbf{x}, t) = x_{k_i}^2$ and $f_j(\mathbf{x}, t) = x_{k_j}^2$ with $i \neq j$, then $g_{i,j} = \frac19$
    \item For $f_i(\mathbf{x}, t) = x_{k_i}^2$ and $f_j(\mathbf{x}, t) = x_{k_j}^2$ with $i = j$, then $g_{i,j} = \frac15$
    \item For $f_i(\mathbf{x}, t) = x_{k_{i_1}}x_{k_{i_2}}$ and $f_j(\mathbf{x}, t) = 1$ with $i_1 \neq i_2$, then $g_{i,j} = \frac{1}{4}$
    \item For $f_i(\mathbf{x}, t) = x_{k_{i_1}}x_{k_{i_2}}$ and $f_j(\mathbf{x}, t) = x_{k_{j}}$ with $i_1,i_2,j$ all distinct, then $g_{i,j} = \frac{1}{8}$
    \item For $f_i(\mathbf{x}, t) = x_{k_{i_1}}x_{k_{i_2}}$ and $f_j(\mathbf{x}, t) = x_{k_{j}}$ with $i_1 \neq i_2 = j_1$, then $g_{i,j} = \frac{1}{6}$
    \item For $f_i(\mathbf{x}, t) = x_{k_{i_1}}x_{k_{i_2}}$ and $f_j(\mathbf{x}, t) = x_{k_{j}}^2$ with $i_1,i_2,j$ all distinct, then $g_{i,j} = \frac{1}{12}$
    \item For $f_i(\mathbf{x}, t) = x_{k_{i_1}}x_{k_{i_2}}$ and $f_j(\mathbf{x}, t) = x_{k_{j}}^2$ with $i_1 \neq i_2 = j_1$, then $g_{i,j} = \frac{1}{8}$
    \item For $f_i(\mathbf{x}, t) = x_{k_{i_1}}x_{k_{i_2}}$ and $f_j(\mathbf{x}, t) = x_{k_{j_1}}x_{k_{j_2}}$ with $i_1,i_2,j_1,j_2$ all distinct, then $g_{i,j} = \frac{1}{16}$
    \item For $f_i(\mathbf{x}, t) = x_{k_{i_1}}x_{k_{i_2}}$ and $f_j(\mathbf{x}, t) = x_{k_{j_1}}x_{k_{j_2}}$ with $i_1 \neq i_2 = j_1$, and $j_2\neq j_1$, then $g_{i,j} = \frac{1}{12}$
    \item For $f_i(\mathbf{x}, t) = x_{k_{i_1}}x_{k_{i_2}}$ and $f_j(\mathbf{x}, t) = x_{k_{j_1}}x_{k_{j_2}}$ with $j_1 = i_1 \neq i_2 = j_2$, then $g_{i,j} = \frac{1}{9}$
\end{itemize}

\section{Illustration of the active learning process}

We provide here two additional examples of the active learning process. The two figures (Figure S1 and S2) are generated in the same manner as Figure 2 in the main paper, although the values of $\gamma$ are changed. In Figure S1, $\gamma = 0.95$, and in Figure S2, $\gamma = 0.05$.

\begin{figure}[H]
    \centering
    \includegraphics[width=\linewidth]{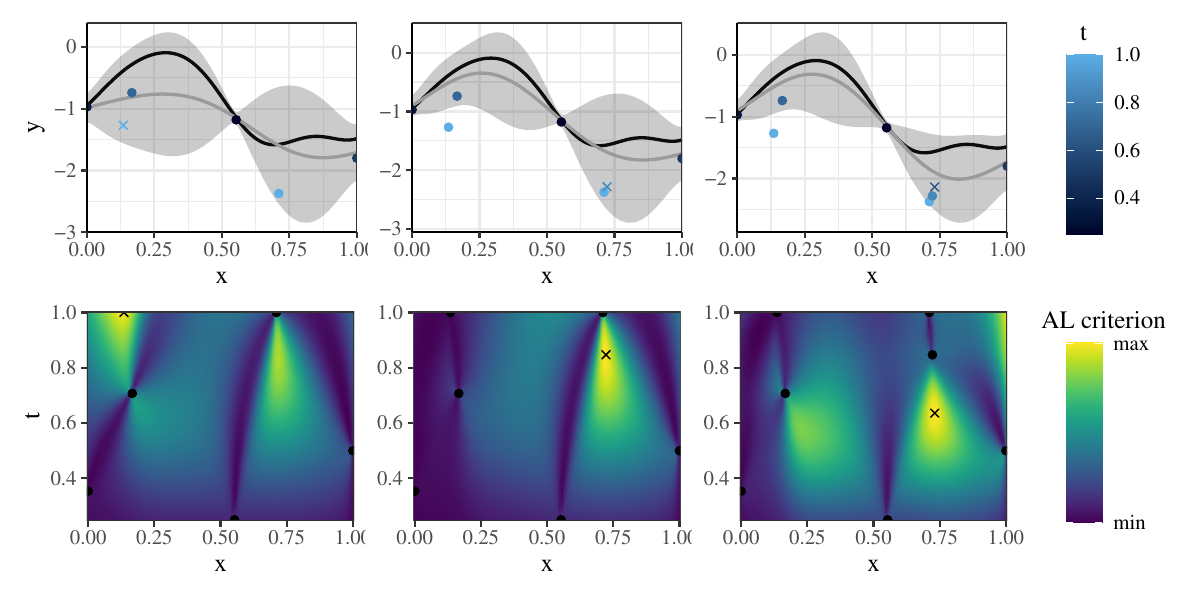}
    \caption{Prediction of our model (top), and active learning criterion surface (bottom). The points represent the current design locations $(\bullet)$, and the best next design point according to the criterion $(\times)$.}
    \label{fig:Al2}
\end{figure}

\begin{figure}[H]
    \centering
    \includegraphics[width=\linewidth]{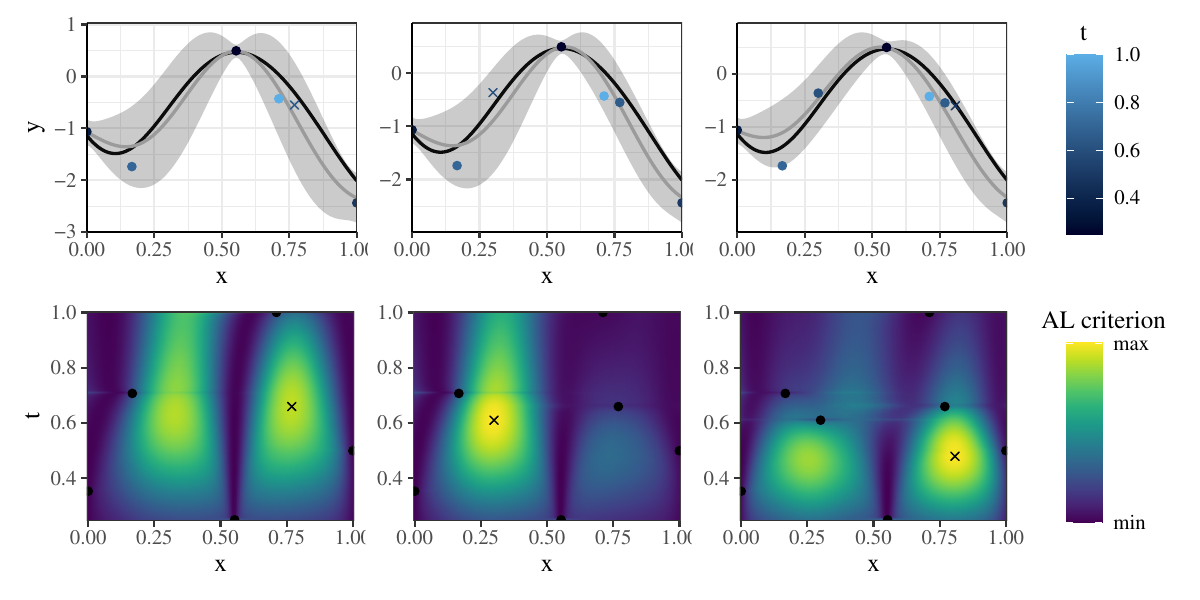}
    \caption{Prediction of our model (top), and active learning criterion surface (bottom). The points represent the current design locations $(\bullet)$, and the best next design point according to the criterion $(\times)$.}
    \label{fig:Al3}
\end{figure}

\section{Definition of RMSE}

For testing the results from a Gaussian process with posterior mean function $\mu$ and posterior standard deviation $\sigma$ against the true function $f$ and with testing data points $\mathbf{X}^{\mathrm{test}} = (\mathbf{x}_i^{\mathrm{test}})_{1\leq i \leq n_{\mathrm{test}}}$, we can define the RMSE as follows,

\begin{align}
    \mathrm{RMSE} & = \left(\sum_i^{n_{\mathrm{test}}} \frac{\left(f(\mathbf{x}_i^{\mathrm{test}}) - \mu(\mathbf{x}_i^{\mathrm{test}})\right)^2}{n_{\mathrm{test}}} \right)^{1/2} \nonumber \\
\end{align}

\section{About initial designs}

We introduce three designs to address these parameter estimation challenges, and compare them with two existing designs that focus on space-filling and projection properties: the Maximum projection design (\texttt{MaxPro}) by \cite{joseph2015maximum}, and the Multi-Mesh Experimental Design (\texttt{MMED}) by \cite{shaowu2023design}, an extension of the \texttt{MaxPro} design in which mesh size samples are chosen based on a weighted maximin criterion that accounts for computational cost.

\begin{figure}[h]
    \centering
    \includegraphics[width=\linewidth]{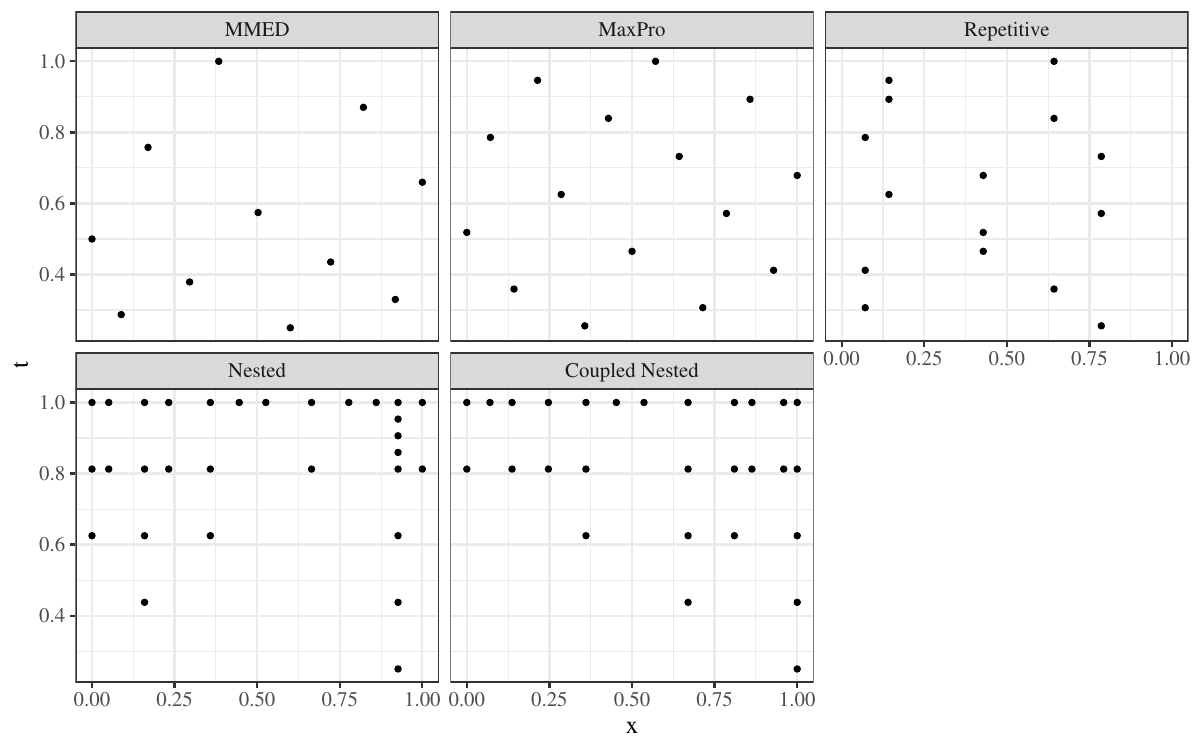}
    \caption{Examples of the different initial design.}
    \label{fig:initial_designs}
\end{figure}

The first experimental design that we introduce derives from the \texttt{MaxPro} design, where we add a constraint on the repetition of data points for each input space location, while retaining the criterion to have maximum projection on the tuning parameter space. We call this design ``Repetitive Design'' (\texttt{Repetitive}). The others two experimental designs are \textit{nested} designs \citep{qian2009nested}, meaning they consist of successive levels from low fidelity to high fidelity, with input space locations iteratively uniformly dropped at a rate of 1/2 at each fidelity level.  The lowest fidelity design is either a \texttt{MaxPro} or space-filling design, and the mesh sizes are evenly distributed across the mesh size space. We consider two variants of this design: a simple nested design (\texttt{Nested}), and another that includes an additional sequence of data points at high tuning parameter values for the same input point (\texttt{Coupled Nested}), which aims at obtaining a better resolution on the short-term behavior of the computer experiments on the tuning parameter space, and thus a more robust estimation of $\gamma$. All the designs are illustrated in Figure \ref{fig:initial_designs}.

Given that the parameter estimation is performed through MLE, using Fisher information provides a reliable method for evaluating parameter estimation performance. For a GP with constant mean and covariance kernel $K_{\boldsymbol{\theta}}$, the Fisher information for parameter $\theta_i$ takes the simple form

\begin{equation}
    F_{\theta_i} = \mathcal{I}_{i,i}(\boldsymbol{\theta}) = \frac12 \mathrm{tr}\Big(\mathbf{K}_{\boldsymbol{\theta}}^{-1}\frac{\partial \mathbf{K}_{\boldsymbol{\theta}}}{\partial \theta_i}\mathbf{K}_{\boldsymbol{\theta}}^{-1}\frac{\partial \mathbf{K}_{\boldsymbol{\theta}}}{\partial \theta_j}\Big)
\end{equation}

where we can use the results from Section B to obtain the closed form expressions of the Fisher information matrix components.\par

We can observe the Fisher information plots parameters $\phi_2^2$ and $\gamma$ across each of their plausible domains (Figure \ref{fig:fisher}). For each parameter, we study the Fisher information with 3 different underlying value of $\gamma$ or $\phi_2^2$, respectively, to observe the possible interaction between the parameter values and the Fisher information.

\begin{figure}[H]
    \centering
    \includegraphics[width=\linewidth]{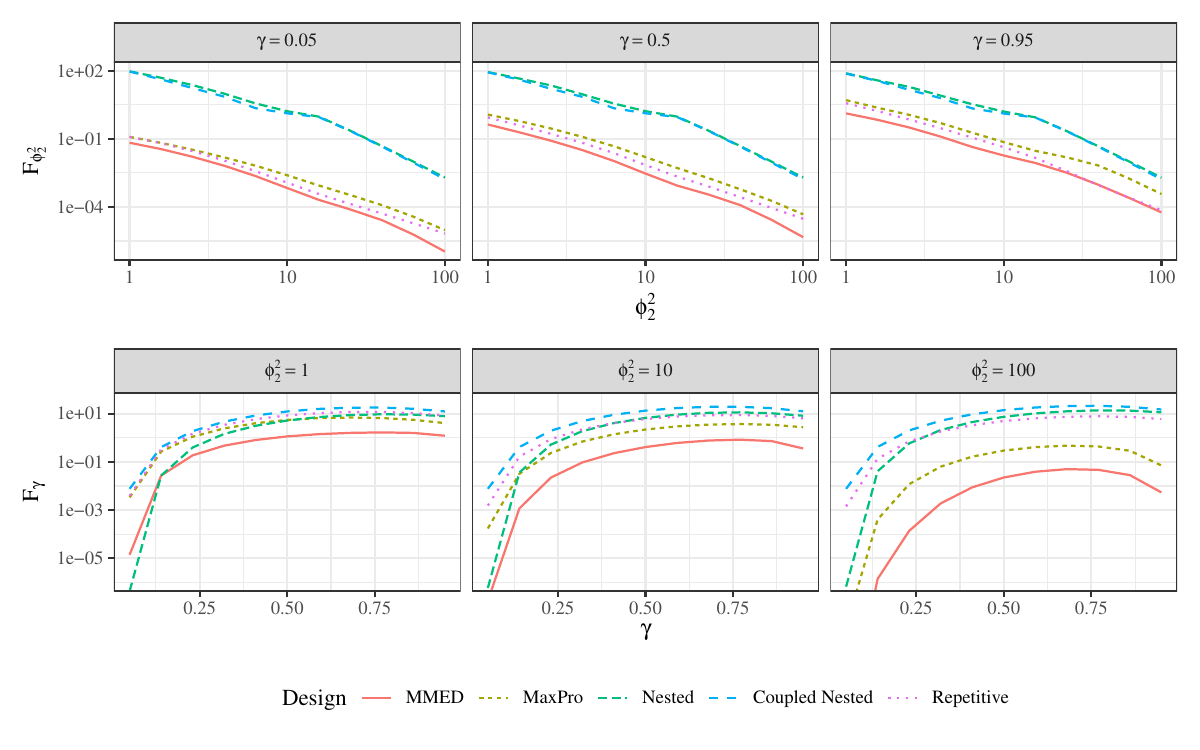}
    \caption{log of the Fisher information for parameters $\phi_2^2$ and $\gamma$ under different $\phi_2^2$ and $\gamma$ values.}
    \label{fig:fisher}
\end{figure}

The \texttt{Coupled Nested} design seems to have the most robust parameter estimation when considering both the Fisher information results for the estimation of $\phi_2^2$ and $\gamma$. It performs all other designs with regards to $\gamma$ while being far superior to all designs except the\texttt{Nested} design with regards to $\phi_2^2$. When considering the impact of the value of $\gamma$ on estimating $\phi_2^2$, we observe that changes $\gamma$ have almost no impact on the Fisher information for the nested designs while significantly impacting the Fisher information of the other design. This observation also applies to the estimation of $\gamma$ when considering different values of $\phi_2^2$. This superior robustness of the nested designs in the parameter estimation stems from the alignment of the data points when projected onto either the input space, or the tuning parameter space. Conceptually, sample paths for a single tuning parameter will become more wiggly as $\phi_2^2$ increases, decreasing the correlation with neighboring points and making the estimation of $\phi_2^2$ harder. This situation deteriorates further for non-nested designs as $\gamma$ decreases since data points that are close on the input space but not on the tuning parameter space become less correlated and share less information as $\gamma$ decreases.\par

We then corroborate the results found in the Fisher information study by constructing 10 datasets generated from our surrogate model for each parameter combination. The initial designs are constructed such that they all incur the same cost, using the cost function $C(t) = t^{-2}$. The results of these two studies are presented (Figures \ref{fig:boxplot_phi2sq_initdes} and \ref{fig:boxplot_gamma_initdes}).

\begin{figure}[H]
    \centering
    \includegraphics[width=\linewidth]{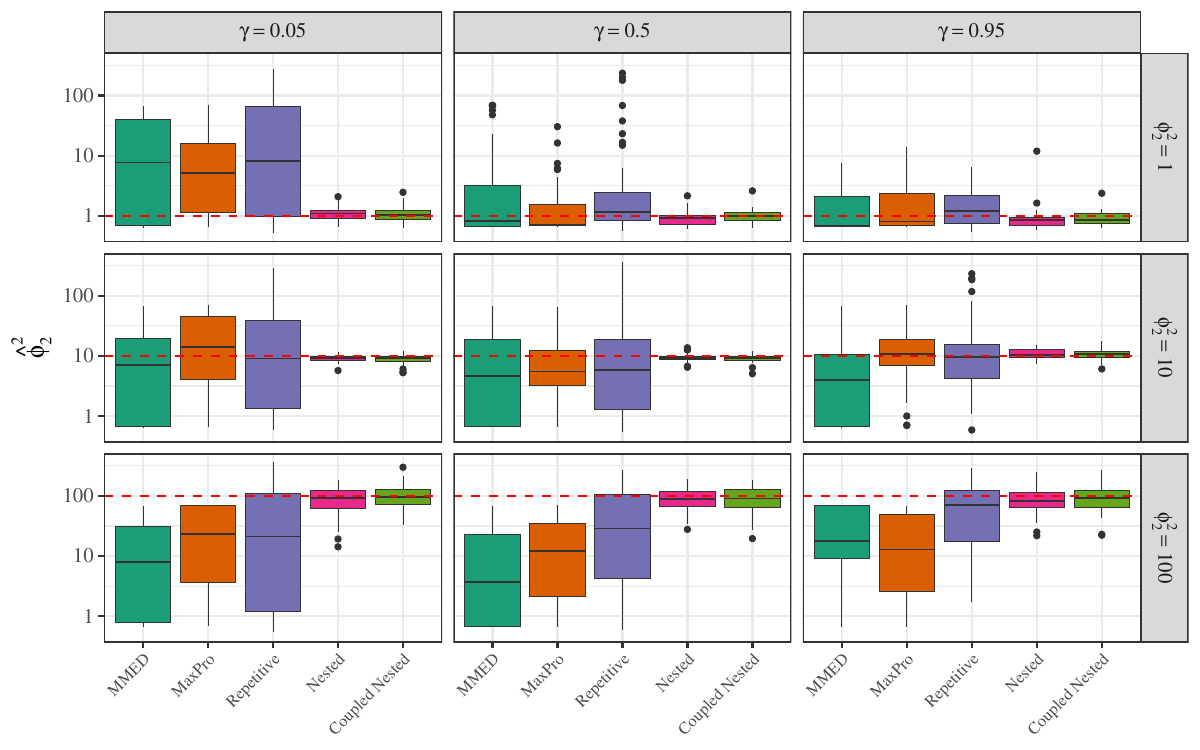}
    \caption{Spread of final estimation of $\phi_2^2$ across 50 repetitions for each combination of true values $\phi_2^2$ and $\gamma$, with respect to the initial design. The true value for each case is indicated by the red horizontal dashed line.}
    \label{fig:boxplot_phi2sq_initdes}
\end{figure}

The results from the simulation studies confirm the findings from the Fisher information study. The nested designs offer superior parameter estimation performance for $\phi_2^2$, while the \texttt{Coupled Nested} design perform better then all the other designs for the estimation of $\gamma$, although the \texttt{Repetitive} design remains competitive in this case. We also observe that the nested design give more consistent parameter estimation as the other parameters change (observing that the performance does not significantly change along different columns) whereas the other designs give significantly different estimation performance as the other parameters change.\par

\begin{figure}[H]
    \centering
    \includegraphics[width=\linewidth]{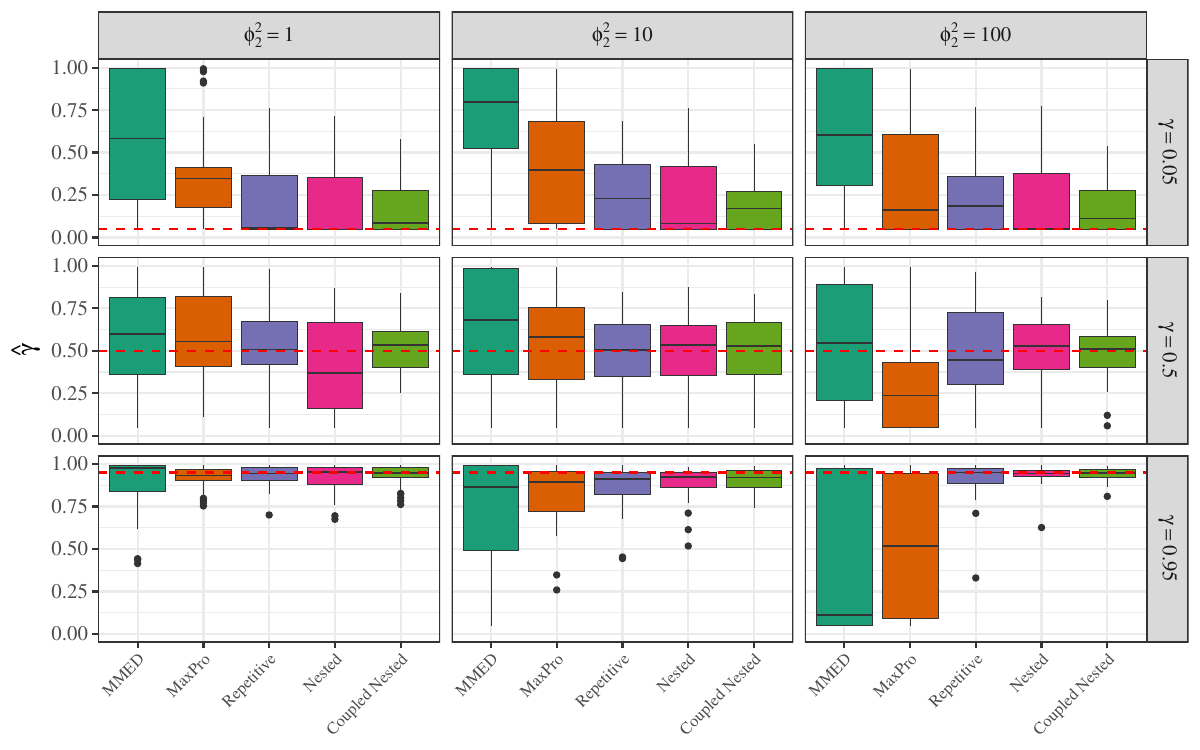}
    \caption{Spread of final estimation of $\gamma$ across 50 repetitions for each combination of true values $\phi_2^2$ and $\gamma$, with respect to the initial design. The true value for each case is indicated by the red horizontal dashed line.}
    \label{fig:boxplot_gamma_initdes}
\end{figure}

Those studies indicate that the nested designs, and in particular the \texttt{Coupled Nested} design, offer better parameter estimation for our multi-fidelity surrogate model. Although important, these findings must be taken further to evaluate this improved parameter estimation can lead to better prediction performances when coupled with our active learning method. 

\begin{figure}[H]
    \centering
    \includegraphics[width=\linewidth]{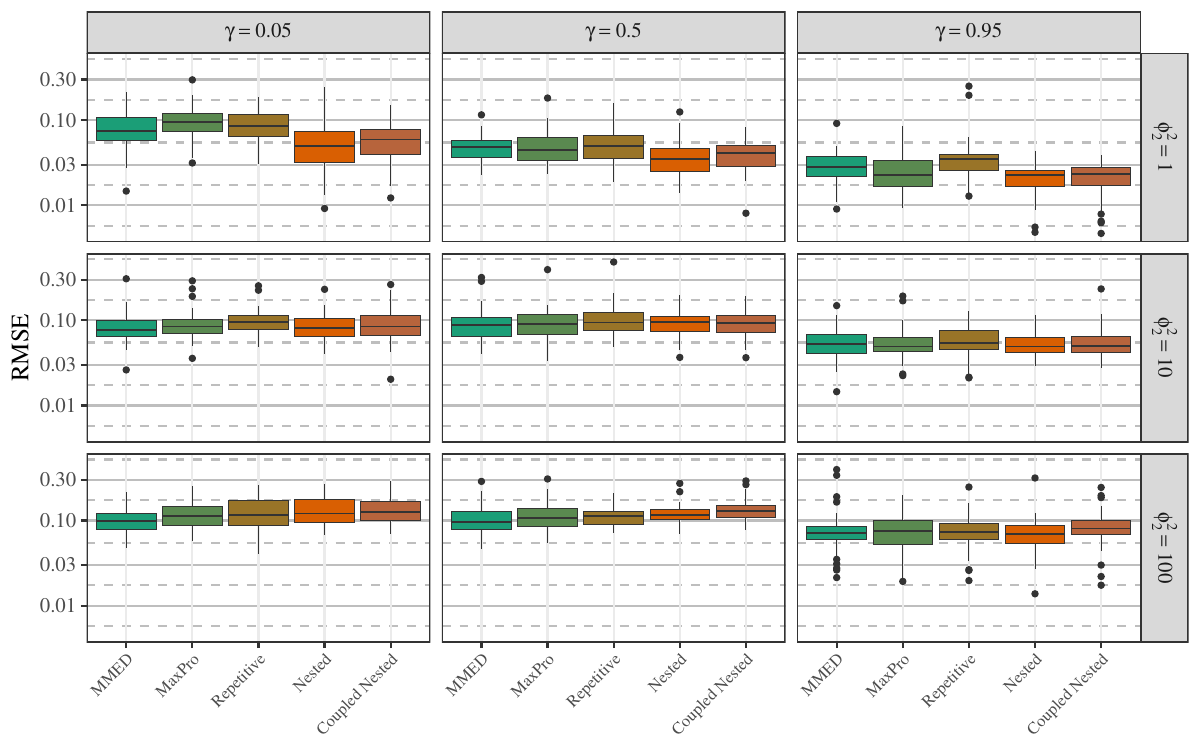}
    \caption{Spread of the final RMSE from the numerical study for each combination of true hyper-parameters and across 50 repetitions for each combination. The labels indicate the initial design used in the active leaning method.}
    \label{fig:sim_rmse_initdes}
\end{figure}

Using the 5 designs as initial designs for our active learning method, we perform the same experiment as described in Section 4 of the main manuscript, and present the results in Figure \ref{fig:sim_rmse_initdes}. We observe that the resulting prediction performance of the active learning method is not significantly affected by the choice of initial design. Indeed, although the nested designs provide more robust parameter estimation early on, their weaker space-filling property when considering its projection on the input space may hinder predictive performance. However, these effects appear to balance out during the active learning process, as we do not observe a significant difference in final predictive performance across different initial designs.





\bibliography{ref}